\documentclass[format=acmsmall]{acmart}
\usepackage{booktabs} 
\usepackage[linesnumbered,lined,boxed,commentsnumbered]{algorithm2e}

\SetAlFnt{\small}
\SetAlCapFnt{\small}
\SetAlCapNameFnt{\small}
\SetAlCapHSkip{0pt}
\IncMargin{-\parindent}

\usepackage{amssymb}
\usepackage{amsmath}
\usepackage{graphicx}
\usepackage{subcaption}
\usepackage{multirow}
\usepackage{rotating}
\usepackage[english]{babel}

\graphicspath{{./}{./Image/}}

\begin{document}
\date{}
\title{Early Routability Assessment in VLSI Floorplans: A Generalized Routing Model}

\author{Bapi Kar}
\affiliation{%
  \institution{Indian Institute of Technology Kharagpur}
  \city{Kharagpur}
  \country{India}}
  
\author{Susmita Sur-Kolay}
\affiliation{
  \institution{Indian Statistical Institute Kolkata}
  \city{Kolkata}
  \country{India}}

\author{Chittaranjan Mandal}
\affiliation{%
  \institution{Indian Institute of Technology Kharagpur}
  \city{Kharagpur}
  \country{India}}

\begin{abstract}

Multiple design iterations are inevitable in nanometer Integrated Circuit (IC) design flow until desired printability and performance metrics are achieved. This starts with placement optimization aimed at improving routability, wirelength, congestion and timing in the design. Contrarily, no such practice exists on a floorplanned layout, during the early stage of the design flow. Recently, STAIRoute \cite{karb2} aimed to address that by identifying the shortest routing path of a net through a set of routing regions in the floorplan in multiple metal layers. Since the blocks in hierarchical ASIC/SoC designs do not use all the permissible routing layers for the internal routing corresponding to standard cell connectivity, the proposed STAIRoute framework is not an effective for early global routability assessment. This leads to improper utilization of routing area, specifically in higher routing layers with fewer routing blockages, as the lack of placement of standard cells does not facilitates any routing of their interconnections.

This paper presents a generalized model for early global routability assessment, HGR, by utilizing the free regions over the blocks beyond certain metal layers. The proposed (hybrid) routing model comprises of (a) the junction graph model in STAIRoute routing through the block boundary regions in lower routing layers, and (ii) the grid graph model for routing in higher layers over the free regions of the blocks. 

Experiment with the latest floorplanning benchmarks exhibit an average reduction of $4\%$, $54\%$ and $70\%$ in netlength, via count, and congestion respectively when HGR is used over STAIRoute. Further, we conducted another experiment on an industrial design flow targeted for $45nm$ process, and the results are encouraging with $~3$X runtime boost when early global routing is used in conjunction with the existing physical design flow.

\end{abstract}


\maketitle

\section{Introduction}
\label{sec:prelim}
Integrated circuit (IC) design for very deep submicron (VDSM) fabrication processes requires multiple iterations, mainly during placement and global/detailed routing or even during post-routing optimization, in order to minimize the routing violations due to several lithographic issues as well as failure to conform to specified performance metrics such as power and speed due excessive wirelength and via count. In the existing physical design (PD) flow \cite{sherw} (see Fig. \ref{fig:pdflow} (a)), global routing \cite{zcao,kast,mcho1,panm1,royj,sherw,xuy,wliu} plays a crucial role for obtaining an acceptable routing solution with minimal/no structural and functional violations. As per this design flow, global routing is traditionally performed, after the placement stage, on a set of nets that define the interconnections between a set of macros and standard cells in the design. 

The global routing model identifies the shortest routing path of a net between the center of a bin to the center of another bin only, which contain the corresponding pins of the net. Hence, no local (intra-bin) routing, connecting a net pin contained in the bin with the bin center, is performed at this stage of the design flow and is pushed to detailed routing stage. This is commonly known as \textit{pin access problem} \cite{chuk}. Few instances of the local/intra-bin routing problem is illustrated in Fig. \ref{intrabin_routing} where the dotted lines in the zoomed-in routing bins denote the missing local/intra-bin routes that were skipped during the global routing stage. This problem may have significant impact on the overall wirelength of the nets, while lack of local congestion estimation within a routing bin and the number of vias required during detailed routing can grossly impact a successful routing closure. Several errors due to design for manufacturibility (DFM) \cite{hchen2} continue to become more critical as the technology nodes gradually shrink, specially due to local routing and congestion hotspots, in addition to severely impacting the speed and power budget. In order to minimize the design iterations, several interleaved global and detailed routing methods such as \cite{bats,jcong,xuy2,ywchang,zhang,zhang3} have been proposed, in addition to integrated placement and global/detail routing frameworks \cite{viswa,panm1,panm2,tlin,liuw} for incremental improvement of the routing solution and faster routing convergence with minimal violation of structural rules and functional specifications.
\begin{figure}[!ht]
\centering
\includegraphics[scale=0.45,angle=90]{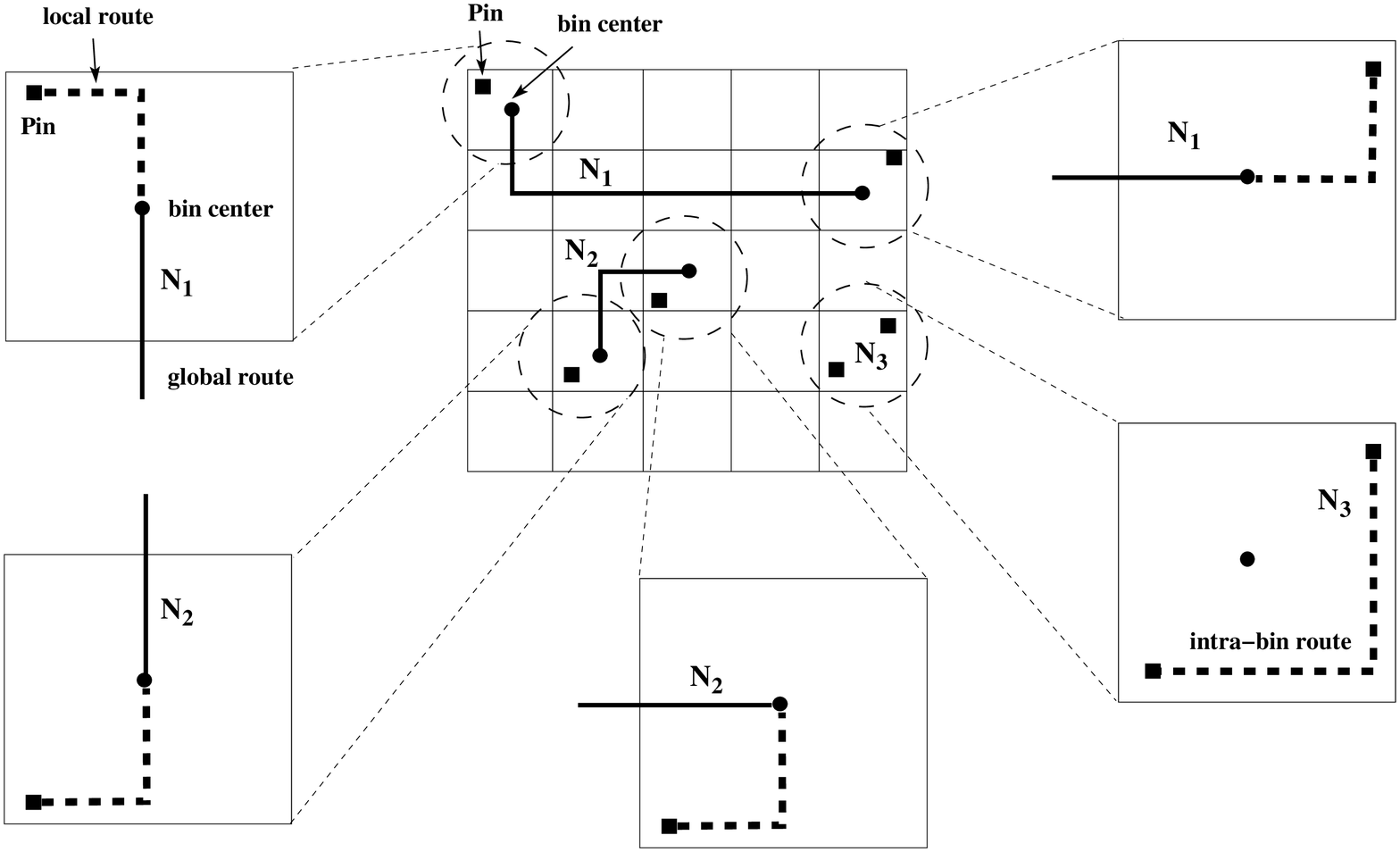}
\caption{Post-placement global routing (solid lines) and intra-bin/local routing (dashed lines)}
\label{intrabin_routing}
\end{figure}

Evidently, all the global routing engines start with an initial placement solution, while iterative placement refinements are done using faster feedback from global/detailed routing engines. No such significant effort has been made for iterative floorplan refinement based on some \textit{early routability estimation} on the floorplans, until recently STAIRoute \cite{karb2} was proposed. This framework provides an insight into the routability of a given floorplan by obtaining an early global routing solution in terms of routing completion, routed wirelength, via count and congestion. These results can be helpful, alike the existing placement and routing framework, in facilitating iterative floorplan refinement as well as guiding the subsequent placement and routing, as per the proposed PD flow illustrated in Fig. \ref{fig:pdflow} (b) (also see how this flow has been adapted to Fig. \ref{fig:olympusflow} in Section \ref{sec:result} for an industrial case study).
\begin{figure}[!ht]
\centering
\includegraphics[scale=0.32]{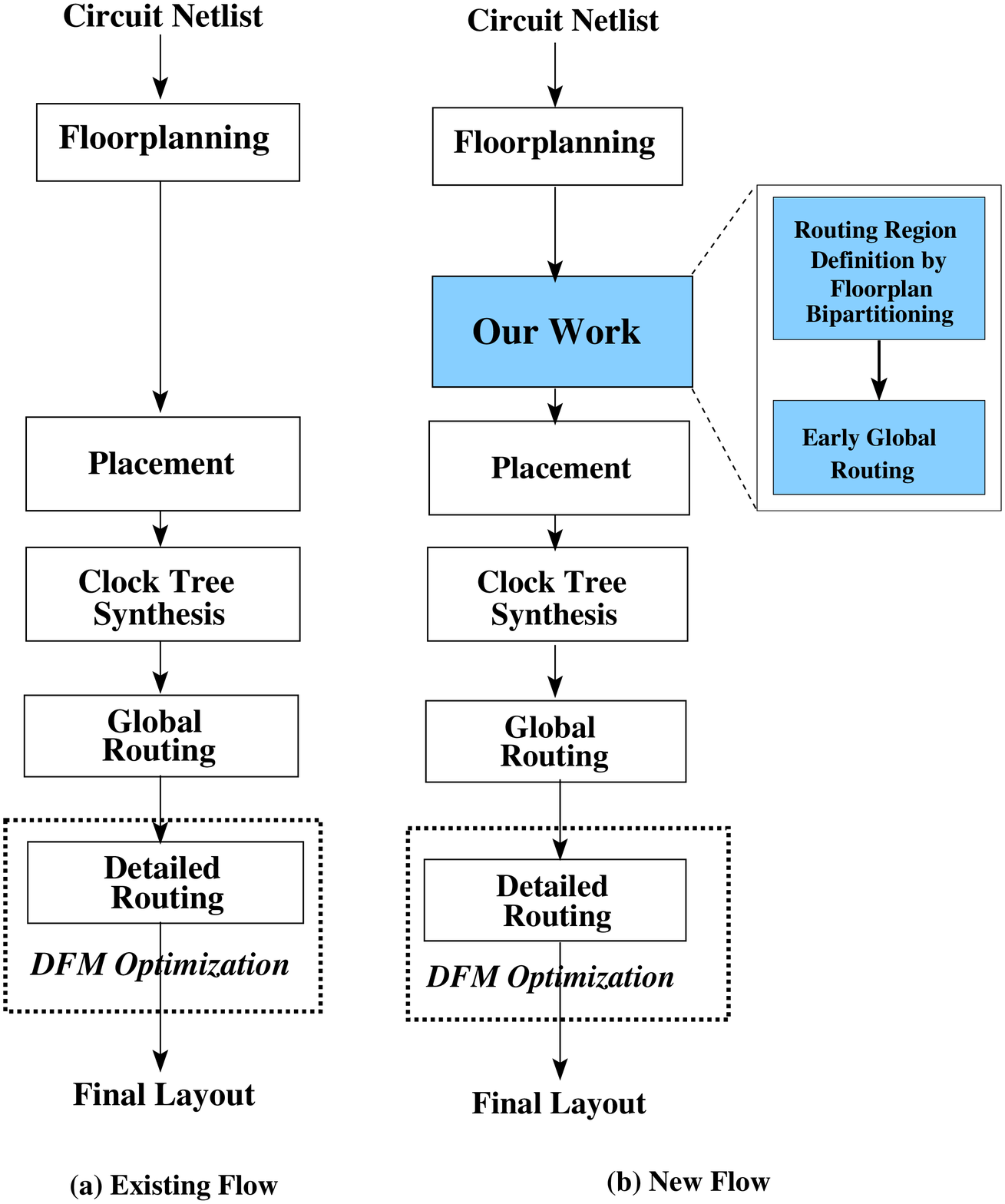}
\caption{VLSI physical design (PD) flow: (a) Existing, and (b) Proposed (aligned with this work)}
\label{fig:pdflow}
\end{figure}

\subsection{Early Global Routing in a Floorplan}
Recently, an effort has been made in assessing early routability of a floorplanned layout by STAIRoute \cite{karb2}, using recursive floorplan bipartitioning \cite{majum1,karb,karb3,karb4} results that identifies a set of monotone staircases in a floorplan as the routing regions. The capacity of these routing regions is obtained from the net cut information of the corresponding node in the bipartitioning hierarchy (MSC tree \cite{karb}). Shortest routing path for a net is identified through these regions assigning the net segments on multiple metal layers depending on the congestion scenario. As cited earlier, these nets were abstracted at the floorplan level and do not account for standard cell connectivity, due to nonavailability of the placement information available for these cells at the floorplanning stage. Instead, these nets define the interconnection between a set of macros (hard blocks) and soft blocks (a cluster of standard cells). In STAIRoute \cite{karb2}, the pin access problem was inherently addressed at the floorplan level, by suitably defining a set of pin-junction edges in the corresponding routing graph (GSRG). These edges facilitate well-defined routing paths between a pin (terminal) of a net to a T-Junction defined in the floorplan \cite{ssk} by the floorplan bipartitioning results.

Although STAIRoute \cite{karb2} is shown to identify a multi-layer routing path of the nets in a floorplan, the corresponding paths are confined in the free regions bounded by the block boundaries only. These regions are identified as monotone staircase routing regions \cite{majum1,karb3}. However, there exist many free routing regions over the macros or the soft blocks (a cluster of highly connected standard cells) over specific routing layers depending on their internal routing, i.e., a block $b_i$ can have free space over a specific routing layer $M_i$ where $M_i$ layers are used for its internal routing. Utilization of these free spaces may improve the routing performance in terms of wirelength, congestion and via count as well as design for manufacturability issues such as \textit{edge placement error} (EPE) \cite{mitra}. At floorplan level, not all the blocks, specially the soft blocks which are basically a cluster of standard cells based on functionality or higher degree of connectedness, come with the number of metal layers used for their respective internal routing. These internal routing for standard cell connectivity are realized traditionally during post-placement routing. Additionally, the lack of placement information of the standard cells during floorplanning leads to an abstraction of the original netlist barring the connectivity to the standard cells. This yields a subset of the original netlist that define the interconnections between a set of macro (hard) blocks and a set of soft blocks only. Moreover, the nets abstracted at the floorplan level are basically much larger nets if all nets are considered that that connect to both macros (including IO pads) and standard cells, as compared to the nets that connect to only the standard cells. Typically the standard cell only nets and also segments of the nets that connect standard cells and macros are masked by the soft blocks. In any case, the early global routing framework always deal with relatively larger nets using multiple routing layers. Therefore, it is reasonable to realize these early methods for assessing the routability of a floorplan which subsequently guides the placement optimization. Intuitively, this early routing results will help in faster convergence of timing driven placement optimization and subsequent global/detailed routing. We see some relevant results of industrial case study presented in Section \ref{sec:result}.

In STAIRoute \cite{karb2}, all the blocks are treated as a set of routing blockages present in all available metal layers. Hence, the routing path of a net is restricted in the monotone staircase regions only, which are located in the adjoining boundary regions of the macros and soft blocks. This routing method uses specified routing layers depending on the vertical/horizontal orientation of the segments in these routing regions. Intuitively, this routing method can be effective when there are fewer number of nets at the floorplan abstraction level interconnecting the macros and soft blocks only, as compared to the original flat netlist seen at the top level that define the connectivity between the macros and standard cells. In practical designs, a macro block usually occupies some, if not all, of the permissible metal layers for routing the nets internal to it; starting from layer $M_1$ up to a layer say $M_j$ below the maximum permissible metal layer $M_{max}$. Therefore, the routing layers \{$M_1 \cdots M_j$\} in the regions over the planar boundary of these macros are blocked for routing the nets, while \{$M_{j+1} \cdots M_{max}$\} layers may be available using horizontal/vertical routing segments. This is also applicable to a soft block due to the routing of its internal nets, apparently blocking a number of metal layers say \{$M_1 \cdots M_k$\} ($M_k < M_{max}$). As cited earlier, floorplan level prior intra-block routing in any soft block is not done due to the lack of placement information of the standard cells in it. Therefore, no \textit{over-the-block} routing is permitted over the soft blocks in these layers. With the increased design complexity, a large number of nets using the same number of routing layers needs to be routed with minimal wirelength and via count during early global routing in a floorplan, over-the-block early global routing of the nets in a floorplan becomes a necessity. 

\subsection{Major Contributions of this work}
In this paper, we present a new early global routing framework called \textit{HGR} in order to realize early global routing of a set of nets in a given floorplan, facilitating \textit{over-the-block} routing of the (sub)nets in the upper metal layers and the routing through monotone staircase routing regions are done in lower layers. As discussed earlier, the internal routing within each block reserve up to a particular routing layer say $M_j$, while HGR utilizes the free space between the boundary regions of the blocks up to $M_j$ and over the blocks above $M_j$ by proposing a novel $3D$ hybrid routing graph. In this graph, the lower layers use the floorplan bipartitioning results, while the upper layers use a modification of the existing grid graph model adapted for floorplans. Like in STAIRoute, this routing model aims to address the pin-access (local/intra-bin) problem by suitably defining a set of relevant edges in this graph. Here, $M_j$ will define the maximum layer already used by all the macros for their internal routing, while the soft blocks will reserve the metal layers up to $M_j$ for their internal routing to be done after the traditional placement of the standard cells is done by the existing design flow (refer to Fig. \ref{fig:pdflow} (b)). Therefore, the routing of the nets internal to the soft blocks defining the interconnection of the standard cells is beyond the scope of this work. 

A comparative example presented in Fig. \ref{hybroute} gives an idea of the proposed work against STAIRoute. This example shows that, in all layers, STAIRoute confines the nets through the staircase routing regions only, while HGR can use both the staircases as well as the free regions above the blocks. It can be noted that if layer $M_j + 1$ and above are uniformly free for routing over all the blocks, the staircase regions prevalent up to $M_j$ cease to exist. In this regard, the grid graph model becomes relevant and hence used in our hybrid routing model.
\begin{figure}[!ht]
\centering
\includegraphics[scale=0.32]{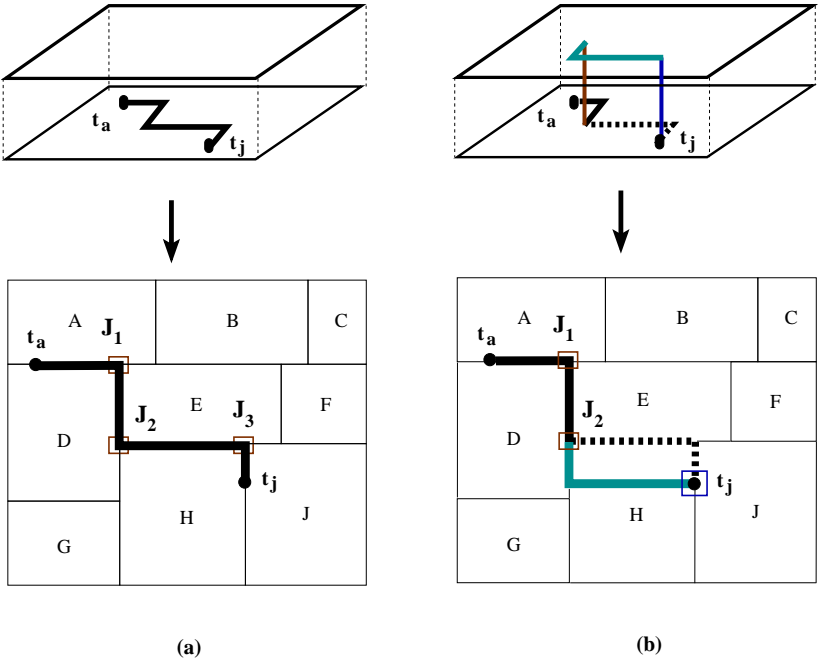}
\caption{Instances of $2$-terminal net ($t_a$, $t_j$) routing using: (a) monotone staircase regions only similar to \cite{karb2}, and (b) both monotone staircase routing in ($M_1$, $M_2$) pair and over-the-block (over $H$ block) routing in $M_3$ and above}
\label{hybroute}
\end{figure}

We organize the rest of this paper as follows: the proposed early global routing approach HGR in the floorplans is described in Section \ref{sec:hgr} for over-the-block routing in a floorplan and also explores the scope of exploring an early abstraction of EPE into the routing penalty for obtaining an EPE-aware early routability assessment method. Experimental results and relevant discussions appear in Section \ref{sec:result}, along with concluding remarks in Section \ref{sec:con}.

\section{This work: A Generalized Model for Early Global Routing}
\label{sec:hgr}
In this section, we discuss the proposed early global routing framework for obtaining \textit{over-the-block} routing of the nets in a floorplan, beyond a certain metal layer $M_j$, while the early routing approach below this layer is similar to that proposed in \cite{karb2}. We consider a scenario when a macro (or soft block) allows early global routing paths for a set of nets over it beyond some specified metal layer say $M_j$, using layer $M_j + 1$ up to $M_{max}$. In this framework, the routing up to $M_j$ is done through the monotone staircase regions similar to STAIRoute i.e. in \{$M_1 \cdots M_j$\} layers while the grid graph model \cite{sherw} is used to obtain the over-the-block routing in the subsequent layers \{$M_{j+1} \cdots M_{max}$\}. In fact, STAIRoute can be seen as a special case of this work when $M_j$ is equal to $M_{max}$, thereby prohibiting the over-the-block routing in any metal layer.

In this work, we adopt a slightly different variant of the existing grid graph model along with the junction graph model presented in STAIRoute \cite{karb2}. The grid graph is overlaid on the junction graph to form a hybrid graph model. Alike in STAIRoute the global staircase routing graph (GSRG) obtained by augmenting the junction graph for each net, this multi-layer hybrid routing graph model is also augmented for each net before identifying the shortest routing path through multiple metal layers using reserved layer model. Our discussion starts with the adoption of the existing grid graph model in this hybrid routing model, followed by the congestion model and the construction of the hybrid routing graph for each net.

\subsection{The Grid Graph Model}
Alike the existing grid graph based routing model, the input layout with the placement details of the standard cells and macros is divided into $m$-by-$m$ global routing bins, where $m$ is predefined. In this paper, we obtain the value of $m$ based on the number of blocks in the floorplan, irrespective of the floorplan topology or the corresponding bipartitioning results, there by removing the dependency of the layout area on the routing efficiency. The value of $m$ is computed as the ceiling of the square root of the number of T-junctions i.e. $\lceil \surd (2n-2)\rceil$, where $n$ is the number of blocks and $2n-2$ is the number of T-junctions in the floorplan \cite{ssk} (also see Lemma $1$ in \cite{karb2}). 

The grid graph $G_g$ = ($V_g$, $E_g$) is defined as follows: each bin corresponds to a vertex $v_p \in V_g$ while each edge $e \in E_g$ denotes a pair of vertices ($v_p$, $v_q$) such that the bins ($g_p$,$g_q$) corresponding to $v_p$ and $v_q$ share a common boundary. Notably, the number of vertices $|V_g|$ and edges $|E_g|$ can be obtained as $m^2$ and $2m(m-1)$ respectively. Hence, both these parameters depend solely on the total number of blocks (macros or soft blocks) $n$ in a given design, not on a particular floorplan topology. 
\begin{lemma}
\label{Lemma1}
For a given floorplan with $n$ blocks, the grid graph $G_g$ can be constructed in $O(n)$ time.
\end{lemma}
\begin{proof}
It is evident that for a layout partitioned into $m$-by-$m$ routing bins, there are $O(m^2)$ vertices and also $O(m^2)$ edges in the grid graph $G_g$. Hence, its construction takes $O(m^2)$, i.e., $O(n)$ time.
\end{proof}

As discussed earlier, the edge capacity in a planar grid graph model is obtained based on the planar routing blockages in the lowest routing layer pair ($M_1$, $M_2$). They are projected on the routing layers beyond $M_2$ based on the technology defined routing track pitch and metal width. In our version of grid graph model, the routing capacity of each edge $e$ is computed based on the fact that if the corresponding boundary between the designated pair of tiles is fully or partially contained within the bounding box of a net $n_i$, it accounts for a capacity of $1$. This is due to the fact that the net can take have a potential routing path through any of these bins, contained within the bounding box of the net \cite{westra,liz}. In this way, the capacity of all the edges is computed for all the nets $N$ = \{$n_i$\} before the routing process starts. The routing of net (segments) using this grid graph model is done by only L/Z shape pattern routing between a pair of pins or junctions or even bin centers in layers beyond $M_2$.

\subsection{The Junction Graph and Congestion Model}
We revisit the junction graph $G_j$ = ($V_j$,$E_j$) defined in \cite{karb2} as below:\\
\qquad $V_j$ = $\{J_p\}$, corresponds to a set of T-junctions, and\\
\qquad $E_j$ = \{\{$J_p$,$J_q$\} $|$ a pair of adjacent junctions $\{J_p, J_q\}$ containing a vertical/horizontal segment $s_k$ of a monotone staircase $C_m$ between them\}. 

The weight of each edge $e \in E_j$ is computed as:
\begin{equation}
\label{Equation2}
wt(e) = length(s_k)/(1-p_e)
\end{equation}
where $p_e$, the \textit{congestion} through the segment $s_k$ between $\{J_p, J_q\}$, is defined as:
\begin{equation}
\label{Equation2a}
p_e = u_e/r_e
\end{equation}

Here $(1-p_e)$ is defined as the \textit{congestion penalty} on the edge weight for routing a net through $e$. As stated before, the reference capacity $r_e$ for a rectilinear staircase segment is computed from the net cut information of the bipartitioning results. Before routing, $u_e$ is set to $0$ and if a net $n_i$ is routed through the corresponding segment $s_k$, then $u_e$ is incremented by $1$. 

A similar routing penalty as per Equation \ref{Equation2} is applied on the edges in the grid graph model adopted in this work. The corresponding length parameter for any edge $e \in E_g$ between a pair of adjacent bins ($g_p$, $g_q$) is denoted as $length_(e)$ and signifies the distance between the center of the bin pair ($g_p$, $g_q$).

\subsection{The Hybrid Global Staircase Routing Graph (hGSRG)}
We define the proposed routing graph by overlaying the junction graph $G_j$ and the grid graph $G_g$ obtained for a floorplanned layout. We call this routing graph as \textit{hybrid global staircase routing graph} (hGSRG) $G_r^i$ = ($V_r^i$, $E_r^i$). For a given net $n_i \in N$ having $t_i$ pins in it, $G_r^i$ is defined as:\\
\qquad $V_r^i$ = $V_j \bigcup V_g \bigcup {t_i}$, and\\
\qquad $E_r^i$ = $E_j \bigcup E_g \bigcup E_{tjg}$\\

Here, $E_{tjg}$ denotes an additional set of edges between (a) a pin and a junction, (b) a pin and a G-Cell, and (c) a junction and a G-Cell and is denoted as:\\
\quad $E_{tjg}$ = $\{t_i,J_k\} \bigcup \{J_k,g_m\} \bigcup \{t_i,g_m\}$\\
where, $\{t_i,J_k\}$ denotes an edge between a net pin $t_i$ and a T-junction $J_k$ in lower layer group ($M_1$, $M_2$) pertaining to the junction graph $G_j$; $\{t_i(J_k),g_m\}$ denotes a vertical edge between a pin $t_i$ (junction $J_k$) and a G-Cell $g_m$ such that $t_i$ ($J_k$) in lower metal layer group ($M_1$,$M_2$) is located within the planar boundary of the bin $g_m$. The edges $\{t_i(J_k),g_m\}$ in this hybrid routing model facilitates the local routing (pin accessibility) within the bin. The steps for the construction of the proposed routing graph hGSRG is illustrated in Fig. \ref{hybridgraph}.
\begin{figure}[!ht]
\centering
\includegraphics[scale=0.4]{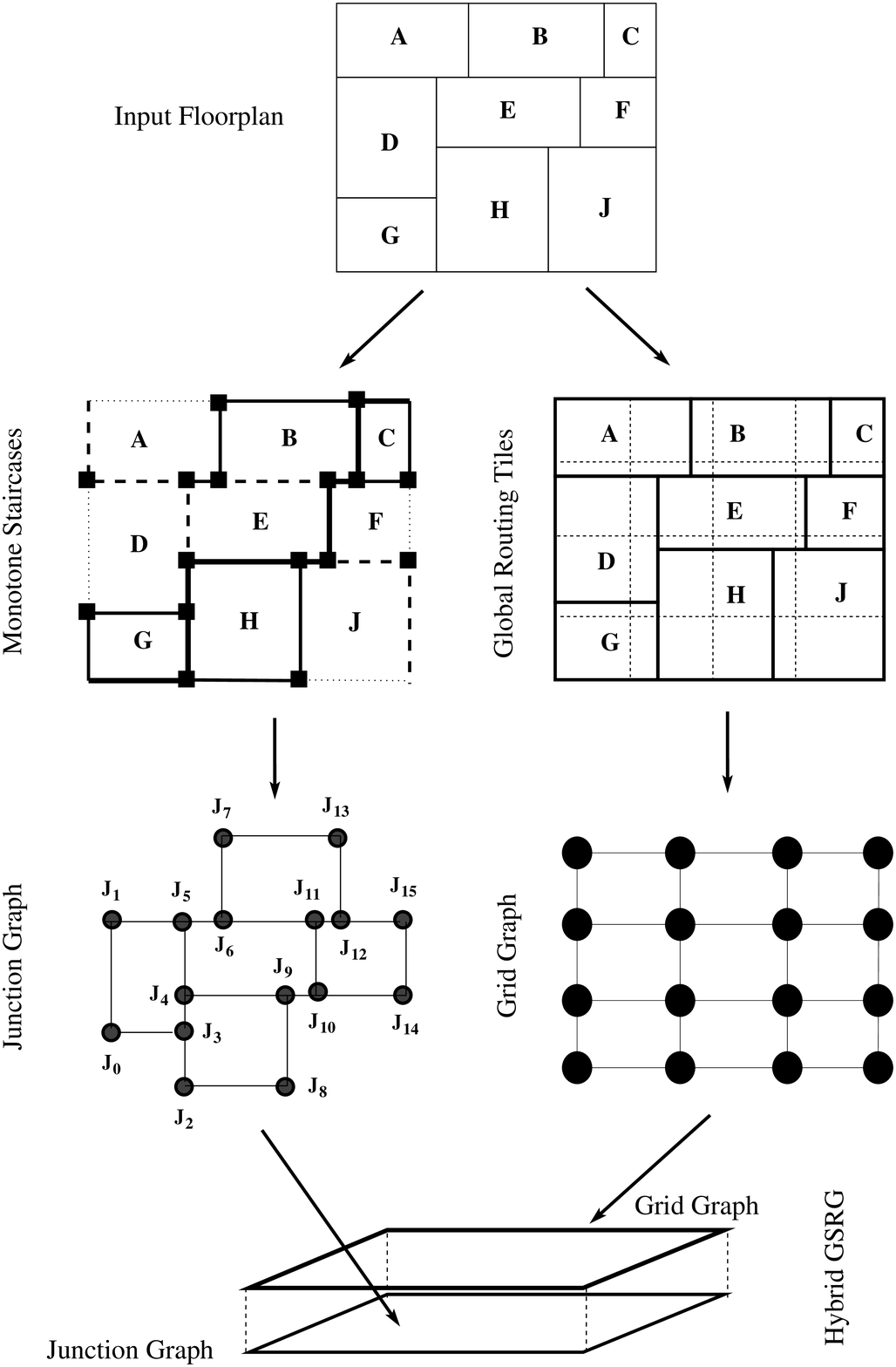}
\caption{Construction of hybrid Global Staircase Routing Graph (hGSRG)}
\label{hybridgraph}
\end{figure}

\subsection{Local/Intra-bin Routing}
Another important aspect of this framework is the routing demand within a bin, i.e., local congestion computation based on the proposed local resource reservation (similar to track reservation \cite{weiy}). Since the routing capacity and demand are the parameters related to bin boundary edges, these local routes are not entitled to utilize them as per the existing grid graph model. On the other hand, our routing model allows us to use grid graph edges for routing the nets beyond some specified routing layers such as $M_3$ beyond the layers used for monotone staircase routing such as ($M_1$,$M_2$) layer pair, as well as obtain the local routes within a bin. The grid graph edges are used to route the net segments that were not routed in ($M_1$,$M_2$), but between the centers of the corresponding bins. These net segments may either be between two junctions or between a pin and a junction. Therefore, the remaining bin center to pin (junction) edges are used to move the net segment to upper layers ($M_3$ and above) with an additional via overhead. As mentioned earlier, all the local routes in this work use L-shaped patterns for minimal via overhead. This is illustrated in Fig. \ref{hybroute_local} (c).
\begin{figure}[!ht]
\centering
\includegraphics[scale=0.45]{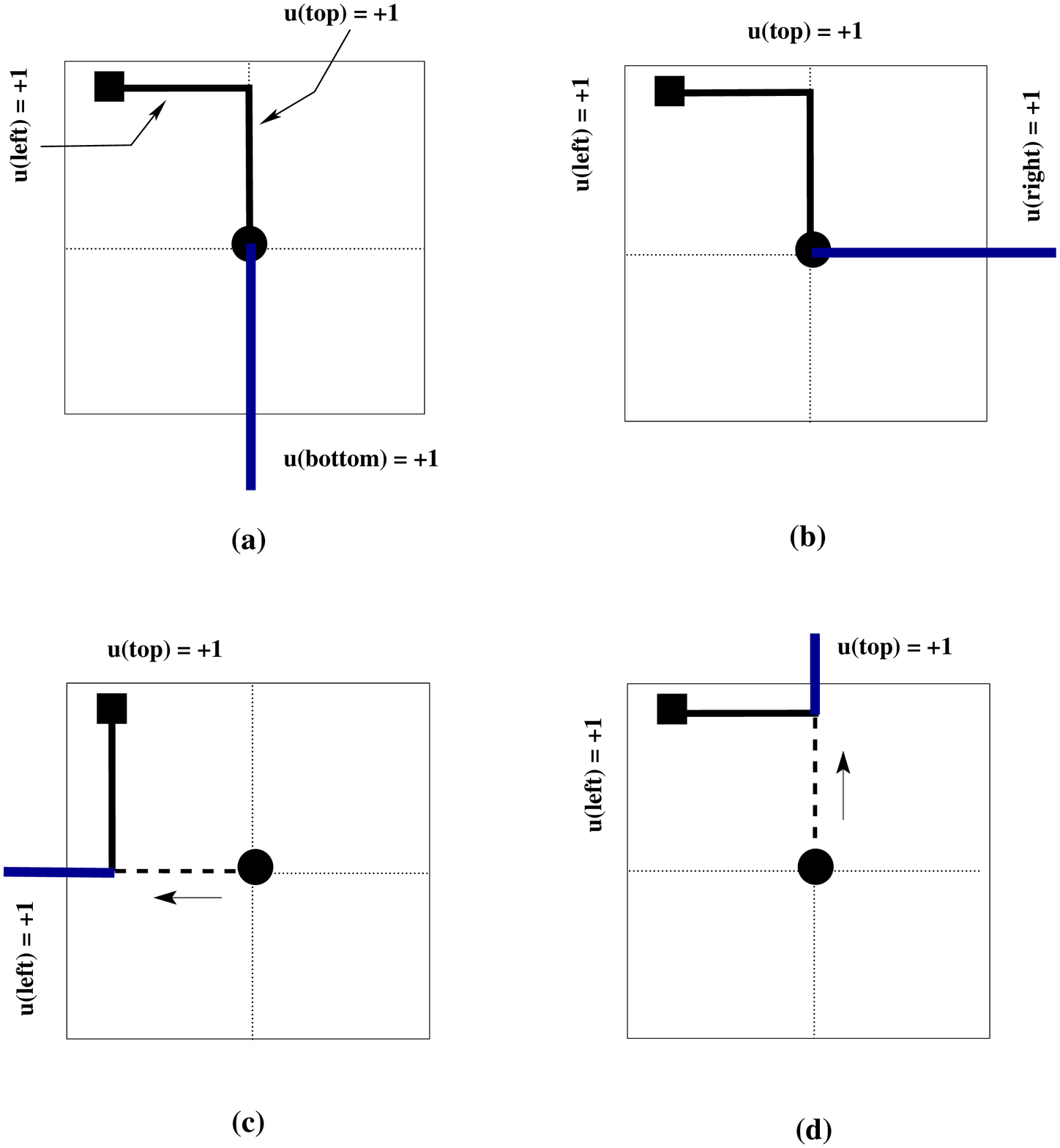}
\caption{Instances of local routing within a bin using routing demand reservation: (a) - (d) enumerates the instances of a fixed pin location vs. bin-to-bin grid routing (in blue)}
\label{intrabin_routing_cap}
\end{figure}

In Fig. \ref{intrabin_routing_cap}, we illustrate the routing demand $u$ for each boundary of a bin after local routing of a net is performed within the bin. This example shows that the position of the pin (junction) in one of the four quadrants of the bin and the global route of the net terminating on the bin center determine the reservation of routing demand $u$ of the corresponding boundary edges; example (a) shows that unit routing demands in the top and left edge are reserved for local routing while bottom edge demand is meant for global route, (b) uses the same like in (a) but with right edge for global route. The example of (c) and (d) are special cases that are closely related to the edge being used for global route with respect to the position of the pin (junction). While the local routing instance in (c) shows that it utilizes the same left edge capacity for both local route and the global route segments and top edge capacity for vertical demand of the local route, (d) depicts a similar case for the global route at the top edge share with vertical segment of the local route, and left edge for horizontal segment of the local route. In these cases, no other edge capacity is relevant and hence not reserved as dictated by zero values in the demand in these edges. In this example, we also note that the wirelength is further minimized (as dictated by the arrow and the dotted line) due to common segment length between global and local routes.

\subsection{Illustration of the proposed routing method}
We illustrate the working of the proposed early global routing method HGR in Fig. \ref{hybroute_local} using a $2$-terminal net ($t_a$, $t_j$). In this example, we assume that this routing method considers only ($M_1$, $M_2$) pair for routing through the monotone staircase regions only, due to the inherent routing blocks due to the macro/soft blocks. As in Fig. \ref{hybroute_local} (a), there are two routing paths between the terminal pair:\\
(i) as per STAIRoute \cite{karb2} obtains ($t_a$, $J_1$, $J_2$, $J_3$, $t_j$) of the net is entirely confined within the monotone staircase routing regions, i.e. through the boundary regions, as denoted partially by black solid line as ($t_a$, $J_1$, $J_2$) and dashed line ($J_2$, $J_3$, $t_j$). This is also illustrated earlier in Fig. \ref{hybroute} (a) by an entirely solid poly-line, by visualizing the dashed poly-line as solid.\\
(ii) here dotted line is used to denote illustrated that this path is congested in lower metal layer pair say ($M_1$, $M_2$) and this framework is capable of finding an alternative path through the higher routing layers (say $M3$ and above), but over the block $H$. Notably, STAIRoute would route the partial routing segment, denoted as the dashed line ($J_2$, $J_3$, $t_j$), in higher layer pairs say ($M_3$, $M_4$). This is due to the fact that its routing model is based only on the Junction graph model which allows routing through the monotone staircase boundary regions only. 
\begin{figure}[!ht]
\centering
\includegraphics[scale=0.40]{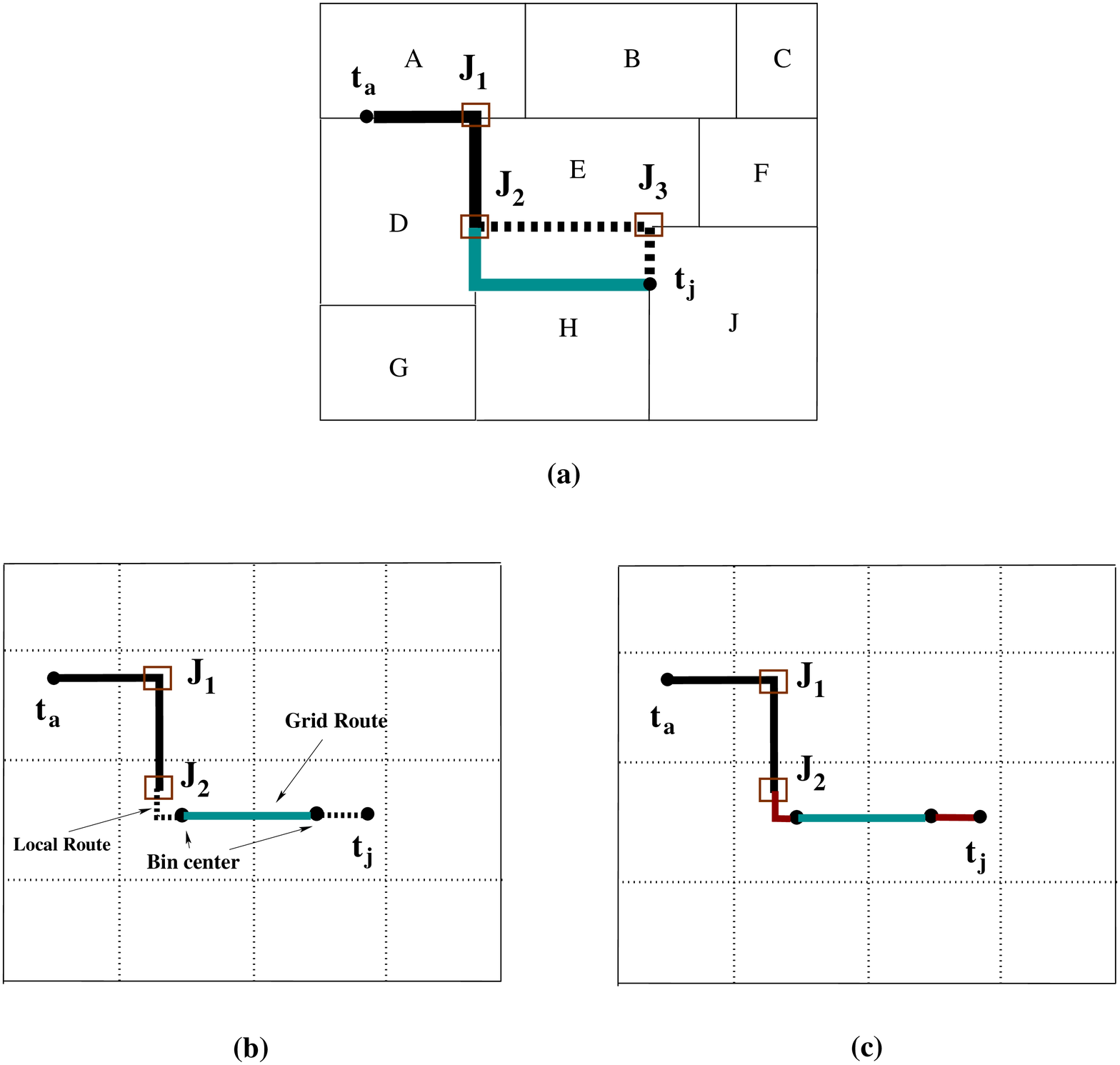}
\caption{Early global routing of $2$-terminal net ($t_a$, $t_j$): (a) monotone staircase routing (partly dotted) and over-the-block routing (over $H$ block), (b) directional routing between bin centers, and (c) local routing for the pin $t_j$ and the T-junction $J_2$ with the corresponding bin centers}
\label{hybroute_local}
\end{figure}

Despite that, both the routing paths incur same wirelength and via count for same number of layer switch, providing more options of alternative routing paths without having any additional routing cost. On the contrary, if STAIRoute is not able to route along ($J_2$, $J_3$, $t_j$) in either $M_3$ or $M_4$ or both due to prevailing congestion scenario, subsequent permissible layer(s) and hence more vias will be used for this interconnection. But HGR will use the free region over the block $H$ in ($M_3$, $M_4$) layer pair only. This implies that fewer metal layers may be used by HGR for overall routing completion as compared to STAIRoute. 

For utilizing the free space over the block $H$, Fig. \ref{hybroute_local} (b) shows how this hybrid routing graph model helps in identifying the remaining routing path ($J_2 \rightsquigarrow t_j$) through the routing bins in upper metal layers ($M_3$, $M_4$) with the help of the grid graph model used in it. Although, the grid graph model obtains a routing path between the centers of a pair of bins, where one bin belongs to $J_2$ while the other contains $t_j$. Till now, the routing between the respective bin centers and $J_2$ ($t_j$) are not done. This is an example of pin-access problem, also commonly known as intra-bin or local routing problem. The third instance in Fig. \ref{hybroute_local} (c) shows the local routing between $J_2$ with one bin center and $t_j$ with the other bin center, both are done in ($M_3$, $M_4$) pair, with the help of the corresponding pin/junction to G-Cell edges (in $\{t_i,g_m\} \bigcup \{J_k,g_m\}$).

\subsection{The Algorithm: HGR}
In Algorithm \ref{Algorithm2}, we summarize the steps for the proposed early global routing method \textit{HGR} (\textit{Hybrid Global Router}). Similar to STAIRoute \cite{karb2}, Dijkstra's shortest path algorithm \cite{cormen} is used to identify the shortest routing path of a $2$-terminal net (segment) in the proposed hybrid routing graph. We also use a similar multi-terminal net decomposition approach in order to identify $2$-terminal net segments as proposed in \cite{karb2}, for the identification of the Steiner tree topology. The set of nets $N$ are ordered first according to net-degree, and then HPWL in the given floorplan based on the bipartition hierarchy. This algorithm takes the junction graph $G_j$ and the grid graph $G_g$ as inputs which are already obtained using the floorplan bipartitioning results obtained by one of the existing works \cite{karb3,karb4} with suitable trade-off values for minimal bend routing. For multi-terminal nets, we use a multi-terminal net decomposition method similar to that proposed in \cite{karb2}. According to this method, for each valid terminal pair, i.e., the valid edge in the resulting spanning tree, we apply this two terminal hybrid routing method, followed by Steiner point identification.

\begin{algorithm}[!ht]
\SetAlgoLined
\SetKwData{Left}{left}\SetKwData{This}{this}\SetKwData{Up}{up}
\SetKwFunction{Union}{Union}\SetKwFunction{FindCompress}{FindCompress}
\SetKwInOut{Input}{input}\SetKwInOut{Output}{output}

\Input{Junction Graph $G_j(V_j$,$E_j$), Grid Graph $G_g(V_g$,$E_g$), an ordered set of nets $N$ = $\{n_i\}$}
\Output{Early global routing of each $t$-terminal ($t\geq2$) net $n_i \in N$}
\BlankLine

Initialize routed (unrouted) net count $c_r$ ($c_u$) to $0$\\
\For{each ordered net $n_i \in N$}{
Construct hybrid GSRG $G_r^i$ for $n_i$\\
\If{$n_i$ is a $2$ terminal net}{
 Identify the shortest routing path using dijkstra single source shortest path algorithm \cite{cormen} on $G_r^i$\\
 \If{There exists a shortest path}{
 Increment $c_r$\\
 Compute netlength, and via count\\
 Update routing demand $u$ for each routing resource along this path\\
 }
 \Else{
 Skip this net and increment $c_u$\\
 }
 }
\Else{
  Construct the node graph $G_c^i$\\
  $T_c^i$ = MST($G_c^i$)\\
  \For{each valid $2$-terminal pairs $(t_j,t_k) \in T_c^i$}{
  Identify the shortest routing path on $G_r^i$\\
  \If{There exists a shortest path for $(t_j, t_k)$ terminal pair}{
  Compute routed netlength, and via-count\\
  Update routing demand $u$ for each routing resource along this path\\
  }
  \Else{
  Skip this net and increment $c_u$\\
  }
  }
  Increment $c_r$\\
  Identify the \textit{Steiner Point(s)} for all two terminal routed net segments\\
  Recompute netlength and via-count for $n_i$\\
 }
 }
Compute routing completion (ratio of $c_r$ and $|N|$) and congestion (see Eqn. \ref{Equation2a}) for all the routing regions across the metal layers.\\
\caption{HGR: An Early Global Router for over-the-block routing in Floorplans}
\label{Algorithm2}
\end{algorithm}

After successful routing of a net $n_i$, the capacity of the corresponding segments in the monotone staircases and grid edges along the routing path of $n_i$ is updated by $1$. In this work, we assume that the pins of a block are located in metal layer $M_1$ or $M_2$. The pins residing in lower metal group ($M_1$, $M_2$) are associated with the junction graph, as shown in Fig. \ref{hybroute_local}, forming pin-junction edges similar to the graph model used in STAIRoute. The pin-bin and junction-bin edges in $G_r^i$ are constructed by identifying the pins/junctions within the corresponding bin, by overlaying the grid graph on the junction graph. The routing cost of these edges, being the vertical connecting edges between $G_j$ and $G_g$, is simply the planner length between the pin/junction location and the center of the bin and the number of vias incurred due to routing through multiple layer groups. 

Alike the existing post-placement global routers, no capacity constraints for these edges are considered as congestion through these vertical edges has little significance, except via overhead. The main contribution in this work is that we use $G_j$ for identifying the (partial) routing path of $n_i$ in lower group of metal layers such as $M_1$-$M_2$ only, while $G_g$ is used to route those nets in the subsequent higher metal layers beyond $M_2$. The routing through the grid graph takes place when a segment of a net cannot be completed in $M_1$-$M_2$ layers, obeying the congestion model in the corresponding rectilinear staircase segments. Although ($M_1$, $M_2$) layer pairs have been used in this work, additional layer pairs such as ($M_3$, $M_4$), ($M_5$, $M_6$) etc. can be used until ($M_{max -1}$, $M_{max}$). When ($M_{max -1}$, $M_{max}$) is also used, this model resembles with STAIRoute leading to no over-the-block early global routing possible in the floorplans.

\begin{theorem}
\label{Theorem1}
Given a floorplan with $n$ blocks and $k$ nets with at most $t$-terminals ($t\geq 2$), HGR takes $O(n^2kt)$ time for finding the routing path of all the nets.
\end{theorem}
\begin{proof}
For a given net $n_i$ with $t$ pins (terminals), there are $O(t)$ pin-junctions edges as per Lemma in \cite{karb2}. Again, for $m$ routing bins, with $t$ pins and $2n-2$ vertices in $G_j$, $O(t)$ pin-bin edges and $O(n)$ junction-bin edges can be obtained. For each net $n_i$, using Lemma $2$, the hybrid GSRG $G_r^i$ construction takes $O(t+n)$, i.e., $O(n)$ time, since $t$ = $o(n)$. 

Alike STAIRoute, our implementation of Dijkstra's single source shortest path algorithm takes $O(n^2)$ time. For $t(>2)$-terminal nets, both the construction of $G_c^i$ and finding its MST require $O(t^2)$ time (see \cite{karb2}). Therefore, finding a routing path for a $t$ terminal net requires $O(n + n^2t + t^2)$, i.e., $O(n^2t)$ time. Hence HGR takes $O(n^2kt)$ time for routing $k$ nets.
\end{proof}

It is evident from Theorem \ref{Theorem1} that HGR presented in Algorithm \ref{Algorithm2} has the same time complexity as that of STAIRoute, with a constant time overhead for the construction of the grid graph and identification of the pin (junction)-bin edges. This also takes into account the routing demand update in each of the edges in the hybrid graph pertaining to both the junction graph and the grid graph, once a net is routed successfully.

\subsection{Early Abstraction of Edge Placement Error}
In this section, we study how edge placement errors (EPE) \cite{dding,mcho2,mitra} occur due to inefficient printability issues of sub-wavelength features using the existing optical illumination system using $193nm$ wavelength. These errors are further aggravated due to the congestion scenario in the routing regions. The intensity map in Fig. \ref{fig-Chap5:epe} (a) depicts that the intensity is not uniform under the mask opening while the same is not zero beyond the mask opening. Therefore, it signifies additional metal width of the wire segment beyond its contour (see Fig. \ref{fig-Chap5:epe} (b)). Thus, if a routing region is more congested, there is little scope to cope up with EPE than doing ripup and reroute for some of the nets (or a part of it), as illustrated in Fig. \ref{fig-Chap5:epe} (b). Moreover, EPE related routing blockage to other nets may leave little room for the detailed routing of the adjacent nets. If this problem is neglected during the detailed routing stage, it will cause a failure during DFM closure stage.

\begin{figure}[!ht]
\centering
\includegraphics[scale=0.75]{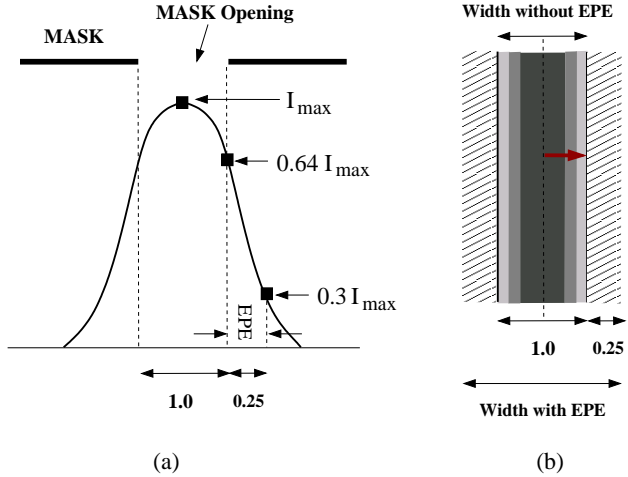}
\caption{Edge Placement Error (EPE): (a) Intensity map vs. mask opening \cite{mitra}, and (b) actual vs. effective metal width with intensity gradient across the normalized width (dark grey to light grey from core to boundary)}
\label{fig-Chap5:epe}
\end{figure}

The intensity map in Fig. \ref{fig-Chap5:epe} (a) depicts the maximum intensity $I_{max}$ at the center of the mask opening \cite{mitra} and the intensity falls off gradually in a pattern similar to $\mathrm{sinc(x)}$ = $sin(x)/x$ function, where $x$ is the distance measured from the center of the mask opening. Notably, the intensity at the mask boundary (edge) is around $64\%$ of $I_{max}$. The intensity gradient across the width (dictated by an arrow in Fig. \ref{fig-Chap5:epe} (b)) signifies the intensity deficiency causing optical proximity errors (OPE). In this model, we consider only the EPE effect in our early routing model in HGR due to nonzero intensity beyond mask opening as the OPE effect due to the said intensity gradient can not be modeled without proper simulation or rule definition at this early stage of physical design flow. According to \cite{mitra}, we consider the threshold point for EPE as the point where the intensity falls to $30\%$ of $I_{max}$, in the region beyond the mask edge, i.e., beyond the wire boundary contour. It amounts to approximately $25\%$ increase in effective metal width on either side of the wire. As a result, it has more interfering effects on the neighboring wire segments of other nets, called EPE induced routing blockage, due to the effect of positive optical interference. This kind of violation due to EPE is discovered during optical rule check (ORC) in the physical verification process of the existing physical design flow. In this case, either wire spreading or rip-up and re-route methods are applied in order to minimize number of such violations. 
\begin{figure}[!ht]
\centering
\includegraphics[scale=0.75]{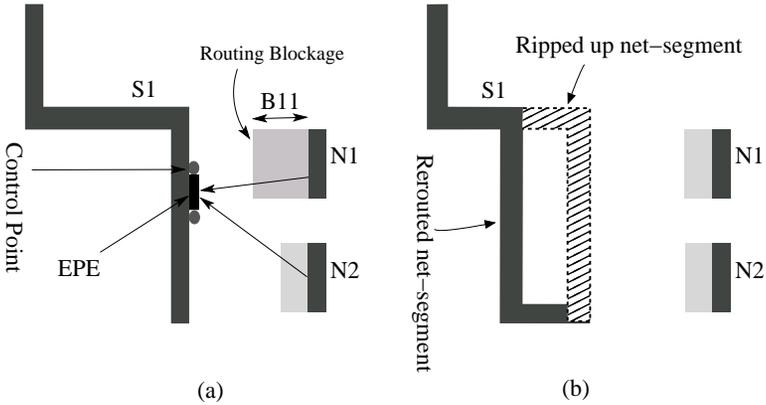}
\caption{Routing blockage due to edge placement error (EPE) \cite{mitra} and Rip-Reroute to alleviate it}
\label{fig-Chap5:epe_blockage}
\end{figure}

In order to incorporate this EPE effect in our early global routing framework HGR, we abstract this effect in the routing demand for assessing the congestion scenario in any metal layer at any given routing instance. This early abstraction of EPE into our proposed routing framework HGR has the potential to reduce the lithography hot-spots due to EPE routing blockage (see Fig. \ref{fig-Chap5:epe_blockage} (a)) at the smaller technology nodes after the detailed routing stage. Therefore, it will reduce the potential overhead of multiple iterations due to wire-spreading or ripup and reroute (RR) during detailed routing (see Fig. \ref{fig-Chap5:epe_blockage} (b)) for EPE hotspor reduction \cite{dding,mcho2,mitra}. In the congestion model of HGR, the penalty due to this EPE cost abstraction is incorporated as follows: after each net is routed through the routing region $e$, its routing demand $u_e$ is incremented by $1.5$ considering the effect of additional $25\%$ metal on the either side. After routing a net $n_{i-1}$, the congestion in the corresponding routing resources along its routing path are computed. The routing graph (hGSRG) $G_{ri}$ for the subsequent net $n_i$ is constructed in order to identify the routing path for it along with the present layer wise congestion scenario in the routing regions. 
\begin{figure}[ht] 
\centering
\includegraphics[scale=0.45]{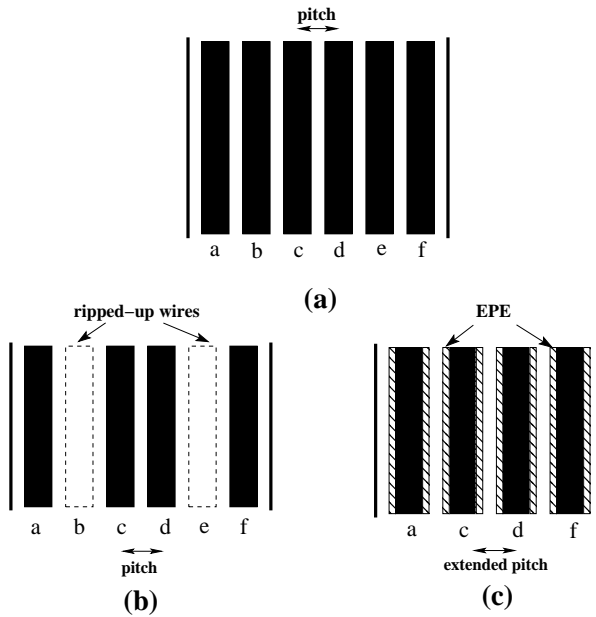}
\caption{Wire spacing/pitch modulation for EPE aware early global routing: (a) all $6$ tracks used, (b) $2$ wires ripped up due to EPE hotspot, and (c) result of EPE aware early global routing}
\label{fig-Chap5:epe_track}
\end{figure}

Due to this process, there will potentially be fewer nets to be ripped up aimed at EPE hotspot reduction during post-routing/layout optimization. The example in Fig. \ref{fig-Chap5:epe_track} (a) showcases this considering six tracks ($a \cdots f$) for routing the nets, without accounting for EPE effect. However, some of the nets, in tracks $b$ and $e$ shown in this example, may be ripped up later during EPE aware routing optimization. On the other hand, as depicted in Fig. \ref{fig-Chap5:epe_track} (c), the proposed EPE aware early routing framework with \textit{extended pitch}, which accounts for taking $u_e$ as $1.5$ for each net routed in region $e$ than using $1.0$ when EPE is not considered, routes only those nets, such as those remained in tracks $a, c, d, f$, after ripping up the other two nets during the existing methods for EPE hotspot reduction. This example shows that a routing solution considering such an early model may potentially reduce the overhead of multiple iterations due to (i) first routing the nets, and (ii) then ripping some of them up in an attempt towards EPE hotspot reduction.

\section{Experimental Results}
\label{sec:result}
In this paper, we used IBM HB floorplanning benchmarks \cite{hben} presented in Table \ref{bench2} for verifying the proposed early global routing method HGR presented in Algorithm \ref{Algorithm2}. These benchmarks were derived from ISPD98 placement benchmark circuits with certain modifications \cite{hben}, in order to form a set of clusters from the standard cells present in the placement benchmark circuits. These clusters not only defined a set of soft blocks, it also helps in defining a subset of the original netlist as the interconnection between the macros and these soft blocks. Notably, these interconnections represent relatively larger nets, some are partial though, barring the standard cell connectivity of the original nets defined in the design. As a result, these modified nets terminate at the boundaries of the macros/soft blocks.
\begin{table}[!ht]
   \caption{HB Floorplanning Benchmark Circuits \cite{hben}}
   \label{bench2}
   \begin{minipage}{\columnwidth}
   \centering
   \begin{tabular}{|c|r|r|r|r|}
      \hline
      Circuit & {\#Blocks} &  {\#Nets} & {Avg.} & {HPWL} \\ 
      Name &   &   & {NetDeg} & {($10^6 \mu$m)} \\\hline      
      \hline
      ibm01 & 2254 & 3990 & 3.94 & 8.98 \\ \hline
      ibm02 & 3723 & 7393 & 4.84 & 22.19 \\ \hline
      ibm03 & 3227 & 7673 & 4.18 & 23.83 \\ \hline
      ibm04 & 4050 & 9768 & 3.92 & 30.82 \\ \hline
      ibm05 & 1612 & 7035 & 5.58 & 18.12 \\ \hline
      ibm06 & 1902 & 7045 & 4.92 & 21.78 \\ \hline
      ibm07 & 2848 & 10822 & 4.44 & 42.48 \\ \hline
      ibm08 & 3251 & 11250 & 4.92 & 46.57 \\ \hline
      ibm09 & 2847 & 10723 & 4.08 & 48.35 \\ \hline
      ibm10 & 3663 & 15590 & 3.85 & 121.23 \\
      \hline
    \end{tabular}
   \end{minipage}
\end{table}

These floorplan instances were generated using \textit{Parquet} \cite{adya,parque} using random seed for each of the circuits, i.e., $ibm01$ to $ibm10$. The proposed algorithm HGR was implemented in C programming language and the experiments were conducted on a Linux platform having Intel Xeon processor running at $2.4$GHz and has $64$GB RAM in it. Since IBM benchmarks \cite{hben} do not provide any pin location details, both STAIRoute and HGR assumed the pins of the blocks connected to the nets at the center of the blocks. A maximum of eight metal layers were permitted to be used for routing the floorplan level nets, using preferred routing directions as horizontal/vertical in odd/even layers. It is also assumed that only two metal layers ($M_1$ and $M_2$ only) out of all the available layers were reserved for internal routing of the macros/soft blocks. For fair comparison, we reran STAIRoute \cite{karb2} with these benchmarks, as the results reported in \cite{karb2} were obtained for much smaller floorplanning benchmark circuits \cite{parque}.

\subsection{Comparison with Existing Early Global Router STAIoute}
First, we present a comparative study between the experimental results obtained by HGR and STAIRoute \cite{karb2}. In this comparison, we consider both with and without abstracted EPE cost in the congestion model (see Eqn. \ref{Equation2}), discussed in the previous section. With these experiments, our aim is study the effectiveness of the proposed generic early global routing framework, which supports over-the-block routing beyond some predefined routing layer, with respect to the existing work STAIRoute \cite{karb2}. In addition to that, we also aimed to incorporate an early abstraction of EPE cost in the routing penalty for both HGR and STAIRoute, as outlined in \cite{karb5} and also presented in the previous section in details. The corresponding results are presented in Table \ref{tab:comp-Results}, for (a) netlength, (b) via count, and (c) average worst congestion \cite{weiy}. 

These results were obtained for $100\%$ routing completion of the nets using up to eight routing layers, ensuring no over-congestion as per the congestion model defined in Eqn. \ref{Equation2}. The results obtained for netlength shows almost identical values obtained by these early global routers for both without and with EPE cost in the routing penalty; HGR shows a slight improvement with an average of $3.5\%$ and $4.5\%$ reduction in wirelength over STAIRoute for both with and without EPE cost respectively. Although both these routers tend to confine the routing paths within the bounding box using one/two/more bends (monotone patterns), the overall netlength varies due to the non-overlapping (non-common) wire segments through the same or different metal layers, as obtained by the multi-terminal net decomposition method (see \cite{karb2} for details). This is apparent as HGR has fewer layer change due to over-the-block routing approach, while STAIRoute uses more layer changes while confining the routing path through the monotone staircase boundary regions only. Accordingly, via count for HGR also shows significant average reduction of $53\%$ and $54.5\%$ respectively over STAIRoute for both the cases, due to fewer layer changes among the wire segments of a net or subnet. 
\begin{table}[!ht]
\scriptsize
  \centering
  \caption{Comparing the Routing Results (with and without EPE cost)}
    \label{tab:comp-Results}
    \begin{minipage}{\columnwidth}
    \begin{tabular}{|l|c|c|c|c|c|c|c|c|c|c|c|c|}
      \hline
      Circuit &  \multicolumn{4}{|c|}{Length($10^6 \mu m$)} &  \multicolumn{4}{|c|}{Via Count($10^5$)} & \multicolumn{4}{|c|}{Worst Avg. Congestion ($wACE4$)} \\  \cline{2-13}
      Name &  \multicolumn{2}{|c|}{ STAIRoute \cite{karb2}} &  \multicolumn{2}{|c|}{HGR}  & \multicolumn{2}{|c|}{ STAIRoute \cite{karb2}} &  \multicolumn{2}{|c|}{HGR} &  \multicolumn{2}{|c|}{ STAIRoute \cite{karb2}} &  \multicolumn{2}{|c|}{HGR} \\  \cline{2-13}
      &  & +EPE & & +EPE & & +EPE & & +EPE &  & +EPE &  & +EPE  \\ \hline
      ibm01 & 11.44 & 11.47 & 11.35 & 11.16 & 2.16 & 2.19 & 1.08 & 1.05 & 0.750 & 0.282 & 0.101 & 0.146 \\ \hline
      ibm02 & 32.18 & 32.27 & 30.73 & 31.26 & 6.17 & 6.28 & 2.82 & 2.79 & 0.917 & 0.294 & 0.100 & 0.135 \\ \hline
      ibm03 & 35.60 & 35.84 & 32.08 & 33.95 & 5.44 & 5.59 & 2.42 & 2.43 & 0.935 & 0.714 & 0.104 & 0.128 \\ \hline
      ibm04 & 39.13 & 39.23 & 39.58 & 37.95 & 7.20 & 7.30 & 3.33 & 3.25 & 0.524 & 0.786 & 0.115 & 0.138 \\ \hline
      ibm05 & 26.43 & 26.45 & 25.94 & 25.73 & 4.01 & 4.06 & 2.08 & 2.11 & 0.203 & 0.277 & 0.108 & 0.158 \\ \hline
      ibm06 & 31.20 & 31.24 & 30.72 & 30.80 & 4.27 & 4.33 & 2.12 & 2.14 & 0.773 & 0.955 & 0.124 & 0.154 \\ \hline
      ibm07 & 56.70 & 56.80 & 56.44 & 54.72 & 7.83 & 7.93 & 3.61 & 3.51 & 0.244 & 0.320 & 0.094 & 0.120 \\ \hline
      ibm08 & 67.81 & 68.06 & 62.78 & 64.64 & 9.35 & 9.52 & 4.09 & 4.21 & 0.647 & 0.569 & 0.082 & 0.117 \\ \hline
      ibm09 & 63.66 & 63.79 & 62.15 & 61.13 & 7.08 & 7.17 & 3.28 & 3.22 & 0.789 & 0.750 & 0.103 & 0.136 \\ \hline
      ibm10 & 153.37 & 154.04 & 141.01 & 133.30 & 10.70 & 11.02 & 4.90 & 4.74 & 0.249 & 0.575 & 0.110 & 0.156 \\ \hline
      Norm. &  &  &  &  &  &  &  &  &  &  &  &  \\
      Geo.  & 1.000 & 1.003 & 0.964 & 0.957 & 1.000 & 1.017 & 0.470 & 0.465 & 1.000 & 0.948 & 0.197 & 0.259 \\
      Mean &  &  &  &  &  &  &  &  &  &  &  &  \\ \hline
    \end{tabular}
    \end{minipage}
\end{table}

In this table, we also present the congestion values measured as the worst average congestion in terms of a parameter called $wACE4$, showing an average reduction of $73\%$ to $80\%$ in case of HGR as compared to that in STAIRoute, in both cases of with and without EPE cost. In this paper, the congestion analysis is motivated by the approach proposed in GLARE \cite{weiy}. The authors in \cite{weiy} defined a parameter called $ACE(x)$ (Average Congestion on Edges) computed for the worst $x\%$ congested edges among all routing layers and the prescribed values of $x$ are one of the values belonging to $\{0.5, 1, 2, 5\}$. In our analysis, we compute an average of $ACE(x)$ for all the prescribed $x$ values and term it average worst congestion $wACE4$. This congestion analysis shows that HGR considers a more realistic approach similar to that adopted to by the post-placement global routing methods, while STAIRoute assesses only congestion in the regions designated as the monotone staircase boundary regions in the floorplan. The congestion values obtained by STAIRoute and hence HGR for lower layers is justified for the lower layer pair ($M_1$, $M_2$) due to the assumption that a very small fraction of of the total layout area is available as the effective routing space, because of the routing space reservation for internal routing in the soft blocks and also the macros whose internal routing has already been done using these two layers. As per this assumption, layers beyond $M_2$ have plenty of free space of the blocks which are effectively used by HGR, while STAIRoute can not use those free space because of its routing model. 
\begin{table}[!ht]
\footnotesize
  \caption{Number of Routing Layers and CPU time (sec) for routing}
    \label{tab:runtime-Results}
    \begin{minipage}{\columnwidth}
    \centering
    \begin{tabular}{|l|c|c|c|c|c|c|c|c|}
      \hline
      Circuit & \multicolumn{4}{|c|}{\#Routing-layers} & \multicolumn{4}{|c|}{CPU time (sec)} \\  \cline{2-9}
      Name  & \multicolumn{2}{|c|}{ STAIRoute \cite{karb2}} &  \multicolumn{2}{|c|}{HGR} & \multicolumn{2}{|c|}{ STAIRoute \cite{karb2}} &  \multicolumn{2}{|c|}{HGR} \\  \cline{2-9}
        & & +EPE & & +EPE & & +EPE & & +EPE \\ \hline
      ibm01 & 6 & 8& 4 & 4 & 2.822 & 2.763 & 10.273 & 10.479 \\ \hline
      ibm02 & 6 & 6& 4 & 4 & 16.359 & 17.144 & 63.529 & 69.391 \\ \hline
      ibm03 & 6 & 6& 4 & 4 & 10.884 & 11.655 & 43.393 & 47.748 \\ \hline
      ibm04 & 6 & 6& 4 & 4 & 21.364 & 23.454 & 90.265 & 81.634 \\ \hline
      ibm05 & 6 & 6& 4 & 4 & 2.799 & 2.841 & 11.375 & 11.502 \\ \hline
      ibm06 & 6 & 8& 4 & 6 & 3.673 & 3.787 & 14.915 & 15.311 \\ \hline
      ibm07 & 8 & 8& 6 & 6 & 12.163 & 12.423 & 51.911 & 50.152 \\ \hline
      ibm08 & 8 & 8& 6 & 6 & 18.102 & 18.087 & 76.065 & 73.402 \\ \hline
      ibm09 & 8 & 8& 6 & 6 & 11.281 & 11.360 & 46.418 & 47.728 \\ \hline
      ibm10 & 8 & 8& 6 & 8 & 26.387 & 27.888 & 126.338 & 126.978 \\ \hline
      Norm. & & &  &  &  &  &  &  \\ 
      Geo.  & 1.000 & 1.029 & 0.700 & 0.749 & 1.000 & 1.032 & 4.114 & 4.157 \\ 
      Mean  & & &  &  &  &  &  &  \\ \hline     
    \end{tabular}
    \end{minipage}
\end{table}

Table \ref{tab:runtime-Results} presents the number of metal layers used out of maximum permissible $8$ layers and the runtime for routing in seconds, for STAIRoute anbd HGR using both with and without EPE costs. These results show that the number of routing layers used for each circuit in case of HGR are $25 - 30 \%$ fewer than that for STAIRoute, considering all the cases. One notable point here is that, for both HGR and STAIRoute, slightly higher number of layers are required in order to obtain $100\%$ when EPE cost is incorporated in the routing penalty. Lastly, as discussed earlier in this paper (see Theorem \ref{Theorem1}), the results also show that HGR needs approximately a constant $4X$ more runtime over STAIRoute to route the same set of nets for the same floorplan instance. 

\subsection{Comparison with Post-placement Global Routers}
In Table \ref{tab-Chap5:Routing-comp}, we present a comparison of the normalized netlength for some of the existing post-placement global routers \cite{mcho1,moffit,ozdal,zhang2,royj,ychang} and the proposed early global routing methods STAIRoute and HGR, based on the corresponding benchmark circuits. Notably, the results for the existing global routers were obtained on IBM ISPD98 placement benchmarks, while HGR and STAIRoute \cite{karb2} obtained the corresponding results on IBM-HB floorplanning benchmarks derived from it \cite{hben}. As cited earlier, a floorplanning benchmark is obtained from a placement benchmark circuit using suitable clustering algorithm, the IBM-HB floorplanning benchmarks were derived from the ISPD98 placement benchmarks \cite{hben}. This conversion incurs significant information loss due to standard cell clustering and the corresponding netlist modification. Therefore, it is unfair to compare the actual netlength obtained for the respective circuits by both the frameworks. Instead, the netlength obtained for each circuit is normalized with respect to the respective Steiner length computed by Flute \cite{cchu}. The results in this table for both STAIRoute and HGR are slightly higher as the nets are routed through the monotone staircase routing regions only for STAIRoute in all the routing layers, while HGR incurred sightly lower netlength value than STAIRoute due to the proposed over-the-block routing approach. However, as we pointed out earlier, the post-placement routers considered all the nets, including the standard cell connectivity, unlike STAIRoute and HGR. Moreover, both STAIRoute and HGR addressed intra-bin/local routing for pin accessibility at the floorplan level, unlike in post-placement global routing which push the pin-accessibility issue to the detailed routing stage. In summary, the main purpose of this study is to understand how these early global routing approaches can assess initial routing metrics like netlength while ensuring $100\%$ routing completion, in terms of the deviation of the overall netlength from the Steiner length without considering any routing blockage. The results show that HGR and STAIRoute deviate $18\%$ and $14.5\%$ from their respective Steiner length, for the corresponding floorplan instances of each circuit. This implies that there is a scope of improvement in the floorplan topology, such as clustering of standard cells and most importantly the locations of the macros/soft blocks), for enhancing the routing performance by using these early routing results in floorplan optimization tools such as Parque \cite{parque}.
\begin{table}
\footnotesize
    \caption{Normalized (w.r.t Steiner length \cite{cchu}) netlength between the existing Global routers and early global routing methods STAIRoute \cite{karb2} and HGR}
    \label{tab-Chap5:Routing-comp}
    \begin{minipage}{\columnwidth}
    \centering
    \begin{tabular}{|l|r|r|r|r|r|r|r|r|}
      \hline
      Circuit & \multicolumn{6}{|c|}{Post-placement Global Routers} &  \multicolumn{2}{|c|}{Early Global Routers}\\ \cline{2-9}
     Name  & \cite{ozdal}$^b$ & \cite{mcho1}$^b$ & \cite{zhang2}$^b$ &  \cite{royj}$^b$ & \cite{moffit}$^b$ & \cite{ychang}$^b$ & STAIRoute\cite{karb2}$^c$ & HGR$^c$\\ \hline
      \hline
      ibm01 & 1.071 & 1.042 & 1.068 & 1.053 & 1.059 & 1.039 & 1.156 & 1.147 \\ \hline
      ibm02 & 1.036 & 1.032 & 1.038 & 1.018 & 1.027 & 1.024 & 1.175 & 1.121 \\ \hline
      ibm03 & 1.007 & 1.007 & 1.007 & 1.005 & 1.010 & 1.005 & 1.175 & 1.059 \\ \hline
      ibm04 & 1.045 & 1.028 & 1.046 & 1.027 & 1.045 & 1.023 & 1.155 & 1.169 \\ \hline
      ibm05$^d$ & - & - & - & - & - & - & 1.198 & 1.176 \\ \hline
      ibm06 & 1.011 & 1.007 & 1.013 & 1.006 & 1.013 & 1.007 & 1.166 & 1.148 \\ \hline
      ibm07 & 1.018 & 1.006 & 1.015 & 1.007 & 1.016 & 1.007 & 1.192 & 1.187 \\ \hline
      ibm08 & 1.005 & 1.008 & 1.009 & 1.006 & 1.010 & 1.006 & 1.197 & 1.109 \\ \hline
      ibm09 & 1.007 & 1.006 & 1.009 & 1.004 & 1.011 & 1.008 & 1.199 & 1.171 \\ \hline
      ibm10 & 1.016 & 1.027 & 1.015 & 1.008 & 1.020 & 1.010 & 1.187 & 1.200 \\ \hline
      Average  & 1.024 & 1.018 & 1.024 &1.015 & 1.024 & 1.014 & 1.180 & 1.149\\ 
     \hline
    \end{tabular}
  \newline
  {\scriptsize $b$ - using ISPD98 global routing benchmarks, $c$ - using IBM-HB floorplanning benchmarks, and $d$ - no result on $ibm05$ of ISPD98 benchmark by the existing global routers}
  \end{minipage}
\end{table}

In the post-placement global routing framework, congestion analysis is done by the number of overflows in the routing edges. On contrary, we used a relative congestion metric (see Eqn. \ref{Equation2a}) inspired by GLARE \cite{weiy}. In this work, we adopted a congestion analysis method using the parameters defined in \cite{weiy} as our congestion model (see Eqn. \ref{Equation2}) does not allow overflow (routing demand exceeding the routing capacity) based computation. Due to this reason, we could not make any direct comparison on congestion overflow for each circuit. Later in this section, we present some overflow based congestion results obtained for some industrial design inputs by an industrial tool which reports both versions of congestion analysis. We also skip the comparison for via counts as both these frameworks are in different scope of operation as per the PD flow. In this comparison, we used only ISPD98 benchmarks whose corresponding IBM-HB floorplanning benchmarks \cite{hben} are the latest ones available for academic research. We are not able to compare the routing results for the latest ISPD07/08 global routing benchmarks, since the corresponding floorplanning benchmarks are not available online. In order to study the effectiveness of the proposed early global routing method HGR with respect to the existing post-placement global routing paradigm, we subsequently present a case study with a well known industrial physical design (PD) tool.

\subsection{An Industrial Case Study}
In order to study the impact of early global routing presented in this work (also in \cite{karb2} on the existing physical design flow, we conducted an industrial case study by developing a framework that can integrate HGR (and also STAIRoute) with industrial PD tool \textit{Olympus-SoC} \cite{olympus}, as outlined in Fig. \ref{fig:olympusflow}. This framework consists of an interface for industry standard LEF/DEF exchange format between STAIRoute/HGR and Olympus tool. In this study, we used a design to be implemented with a $45nm$ physical design library \cite{nangate_45nm}, while the physical verification was done by Calibre tool suite \cite{calibre}.
\begin{figure}[!ht]
\centering
\includegraphics[scale=0.33]{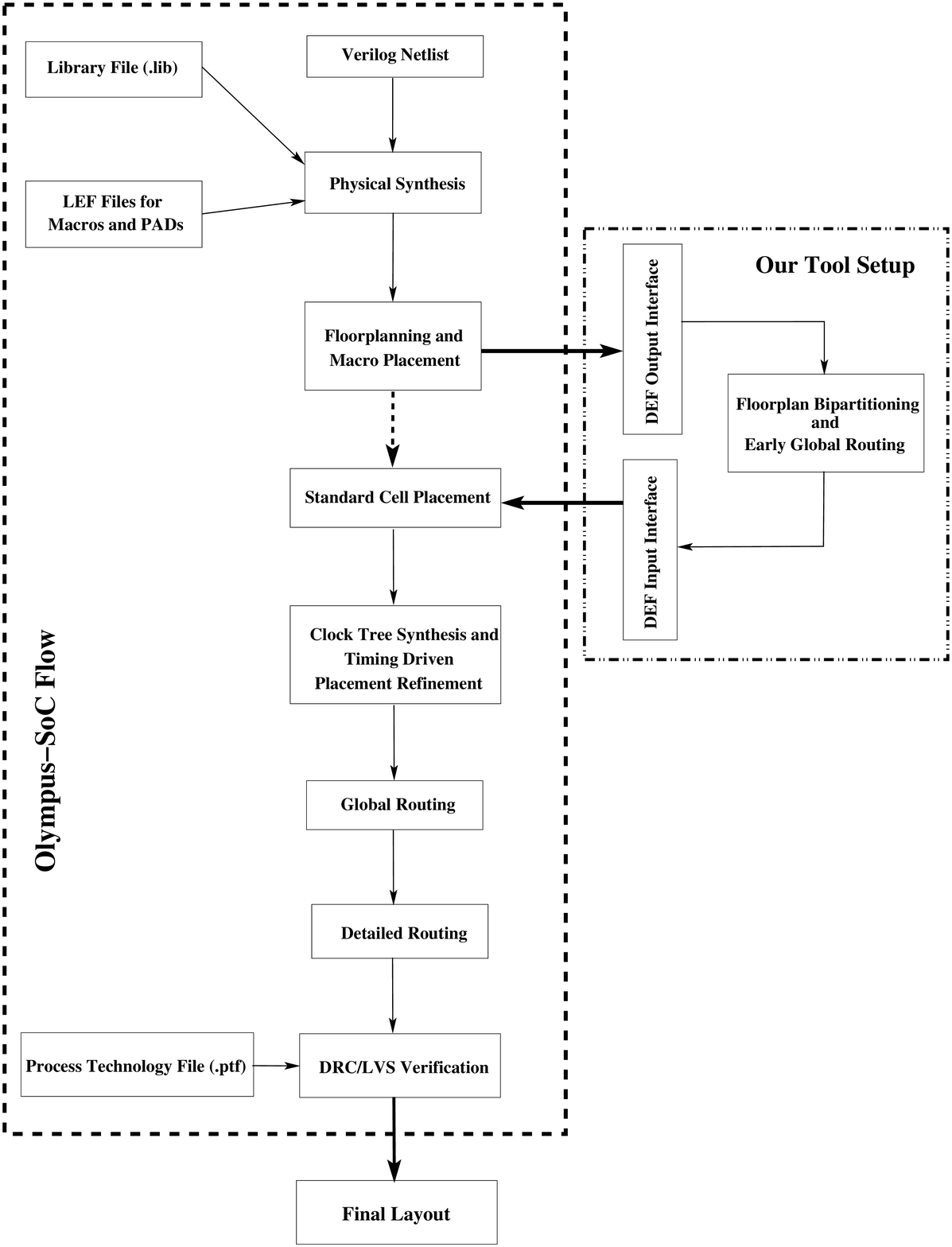}
\caption{Modified Olympus-SoC \cite{olympus} Physical Design (PD) flow with integrated Early Global Routing (HGR and STAIRoute \cite{karb2})}
\label{fig:olympusflow}
\end{figure}

In this exercise started with a floorplan instance of a design generated by Olympus-SoC tool. Default placement constraints were used for the subsequent placement engine and the design utilization target was set at $60\%$. Subsequent stages in Olympus implemented the design as per Fig. \ref{fig:olympusflow}, having no early routability assessment done on the floorplan. In another attempt, STAIRoute/HGR obtained the same floorplan instance through the LEF/DEF interface and obtained the early global routing solution for randomly chosen ($\gamma$, $\beta$) values such as ($0.5$, $0.2$) used to obtain the recursive floorplan bipartitioning results \cite{karb3} and define the routing regions. For routing, we used $10$ metal layers using preferred routing directions in each layer, horizontal being on the odd layer and vertical on the even, as supported by the given technology. In this study, we emphasized on timing driven placement and hence on parameters like TNS/WNS. Using both the approaches in Olympus, i.e. with and without early global routing, we obtained the final layout with no DRC/LVS violations, and also ensured zero timing violations in terms of non-negative WNS/TNS values. We used all other default values pertaining to Olympus-SoC tool for this exercise and noticed that timing was given more priority, than congestion/via minimization, using buffer insertion. This is the reason that Olympus tool used more buffer/inverter cells when STAIRoute was used. Hence more number of nets were introduced during timing driven standard cell placement, after early global routing (see Fig. \ref{fig:olympusflow} (b)). In all the cases, the target of $60\%$ placement density (utilization \%) was met while no timing violations were reported, and non-negative WNS/TNS values were achieved subsequently. It is important to note here that, the present day design philosophy goes for fixed outlined layout with predefined layout area. Therefore, higher utilization value may sometime lead to infeasible, sometimes incomplete, routing solution due to either (i) maximum usage of higher number of layers, and (ii) even the maximum number of layers permissible by the technology is not sufficient.
\begin{table}
\footnotesize
    \caption{Impact of early global routing on Olympus PD flow \cite{olympus}}
    \label{tab-Chap6:Study_with_egr_hgr_2}
    \begin{minipage}{\columnwidth}
    \centering
    \begin{tabular}{|c|r|r|r|}
      \hline
       \textbf{Parameter(s)} & \textbf{Olympus-SoC \cite{olympus}} & \textbf{Olympus \cite{olympus} } & \textbf{Olympus \cite{olympus}} \\
       & & + STAIRoute \cite{karb2} & + HGR \\
      \hline
      Standard Cells (Macros) & 16360 (4) & 16611 (4) & 16504 (4) \\ \hline
      Nets & 16474 & 16725 & 16618 \\ \hline
      Buffer/Inverter & 225/6527 & 225/6575 & 225/6441 \\ \hline
      Placed Area (Buf area) ($\mu m^2$) & 45909 (10104) & 45966 (10157) & 45834 (10018) \\ \hline
      Utilization (\%) & 59.17 & 59.24 & 59.07 \\ \hline
      Wirelength(mm) & 273.02 & 273.52 & 273.43 \\ \hline
      Via Count (x$10^3$) & 52.80 & 53.13 & 52.84 \\ \hline
      WNS/TNS (ns) & 0.00/0.00 & 0.00/0.00 & 0.00/0.00 \\ \hline
      Avg. Congestion (X/Y) & 0.123/0.111 & 0.122/0.112 & 0.122/0.112 \\ \hline
      Worst Congestion (X/Y) & 0.826/0.928 & 0.791/0.916 & 0.782/0.909 \\ \hline
      Edge Overflow (X/Y) & 0/0 & 0/0 & 0/0 \\ \hline
      DRC/LVS violations & No & No & No \\ \hline
      \# of Routing Layers & 8 & 8 & 8 \\ \hline
      CPU time (sec) & 1979 & 688 & 699\\ \hline
    \end{tabular}
    \end{minipage}
\end{table}

The corresponding results are presented in Tab \ref{tab-Chap6:Study_with_egr_hgr_2} for the given floorplan instance. For the results pertaining to early global routing, important observations can be made in worst/average congestion (in X/Y direction), via count, runtime and even timing despite insignificant increase (less than $1\%$) in cell and net count due to timing driven placement of standard cells by Olympus with early global routing than in standalone Olympus mode. Moreover, no DRC/LVS violations nor placement density constraint violations (utilization $\ngtr 60\%$) were noticed. Runtime for Olympus with STAIRoute/HGR is significantly improved due to fewer number of iterations in order to converge on an acceptable routing solution with no timing violations, being roughly three times faster than that for standalone Olympus, while results for HGR+Olympus are shown to be the best among all. Moreover, significant runtime improvement for the entire flow can also be noticed due to fewer number of iterations, as the timing driven placement optimization appears to have been guided by the early global routing results for both HGR and STAIRoute.

\section{Conclusion}
\label{sec:con}
In this paper, we present a new early global routing framework that facilitates over-the-block early global routing of the nets in a given floorplan. In this work, layer $M_2$ is assumed to be the maximum routing layer used for internal routing for macros/soft block. While STAIRoute \cite{karb2} identifies the routing paths of the nets through the monotone staircase routing regions in all routing layers, the proposed method HGR routes the nets through these regions only in lower metal layers ($M_1$, $M_2$). For higher routing layers, i.e. beyond $M_2$, HGR uses the free space over the blocks for routing the nets or parts of the nets. Alike STAIRoute, this model also attempts to address the pin access problem at the floorplan level by suitable edge definitions in the proposed hybrid routing graph, as presented in Section \ref{sec:hgr}. Therefore, HGR can be seen as a generic version of STAIRoute. If the maximum layer assumed to be reserved for internal routing is taken as the maximum allowable routing layer, then there will be no free space of over the macro/soft blocks. This implies that all the nets are to be routed through the monotone staircase regions situated at the boundary regions of each of the blocks, and applicable for all layers.

Experimental results presented in Section \ref{sec:result} show that HGR obtains smaller wirelength, fewer via count and also less congestion as compared to STAIRoute, with a more realistic early global routing approach by obtaining over-the-block routing paths beyond a specific routing layer. These results also ensure that the congestion values are constrained to $100\%$ as per the proposed congestion model, with $100\%$ routability. A comparison with the existing post-placement global routers also show that, even without detailed placement of standard cells, the proposed early global routing approach HGR can have a good idea of routability, wirelength and congestion values in a given floorplan instead of using primitive metrics like half perimeter wirelength (HPWL) and pre-routing probabilistic congestion analysis. The industrial case study also indicates the potential of early global routing framework in guiding acceptable final routing solutions in fewer iterations, and with fewer/no violations due to design for manufacturibility issues and also leverage on minimalistic compromise on the performance metrics like speed, within a competitive time frame.

Based on the results from the industrial case study, we plan for a new placement and routing framework guided by this generic early global routing method for a given floorplan solution and subsequently incorporate the DFM issues with enhanced models using extensive simulations results for smaller VDSM process nodes. We also plan to explore the scope of this work in $3D$ IC design flow.


\renewcommand\bibname{References}

\begin{thebibliography}{99}
\ifx \showCODEN    \undefined \def \showCODEN     #1{\unskip}     \fi
\ifx \showDOI      \undefined \def \showDOI       #1{#1}\fi
\ifx \showISBNx    \undefined \def \showISBNx     #1{\unskip}     \fi
\ifx \showISBNxiii \undefined \def \showISBNxiii  #1{\unskip}     \fi
\ifx \showISSN     \undefined \def \showISSN      #1{\unskip}     \fi
\ifx \showLCCN     \undefined \def \showLCCN      #1{\unskip}     \fi
\ifx \shownote     \undefined \def \shownote      #1{#1}          \fi
\ifx \showarticletitle \undefined \def \showarticletitle #1{#1}   \fi
\ifx \showURL      \undefined \def \showURL       {\relax}        \fi
\providecommand\bibfield[2]{#2}
\providecommand\bibinfo[2]{#2}
\providecommand\natexlab[1]{#1}
\providecommand\showeprint[2][]{arXiv:#2}

\bibitem[\protect\citeauthoryear{??}{cal}{[n. d.]}]%
        {calibre}
 \bibinfo{year}{[n. d.]}\natexlab{}.
\newblock \bibinfo{title}{{Calibre Tool Suite}, {Mentor Graphics Inc.}}
\newblock   (\bibinfo{year}{[n. d.]}).
\newblock
\showURL{%
\url{https://www.mentor.com/products/ic_nanometer_design/verification-signoff}}


\bibitem[\protect\citeauthoryear{??}{hbe}{[n. d.]}]%
        {hben}
 \bibinfo{year}{[n. d.]}\natexlab{}.
\newblock \bibinfo{title}{{IBM-HB Floorplanning Benchmarks}}.
\newblock   (\bibinfo{year}{[n. d.]}).
\newblock
\showURL{%
\url{https://cadlab.cs.ucla.edu/cpmo/HBsuite.html}}


\bibitem[\protect\citeauthoryear{??}{nan}{[n. d.]}]%
        {nangate_45nm}
 \bibinfo{year}{[n. d.]}\natexlab{}.
\newblock \bibinfo{title}{{NanGate} $45nm$ {OpenCell Library}}.
\newblock   (\bibinfo{year}{[n. d.]}).
\newblock
\showURL{%
\url{http://www.nangate.com/?page_id=2325}}


\bibitem[\protect\citeauthoryear{??}{oly}{[n. d.]}]%
        {olympus}
 \bibinfo{year}{[n. d.]}\natexlab{}.
\newblock \bibinfo{title}{{Olympus-SoC} tool, {Mentor Graphics Inc.}}
\newblock   (\bibinfo{year}{[n. d.]}).
\newblock
\showURL{%
\url{https://www.mentor.com/products/ic_nanometer_design/place-route/olympus-soc}}


\bibitem[\protect\citeauthoryear{??}{par}{[n. d.]}]%
        {parque}
 \bibinfo{year}{[n. d.]}\natexlab{}.
\newblock \bibinfo{title}{{Parquet} Floorplanner and MCNC/GSRC Floorplanning
  Benchmarks}.
\newblock   (\bibinfo{year}{[n. d.]}).
\newblock
\showURL{%
\url{https://vlsicad.eecs.umich.edu/BK/parquet}}


\bibitem[\protect\citeauthoryear{Adya and Markov}{Adya and Markov}{2003}]%
        {adya}
\bibfield{author}{\bibinfo{person}{S.~N. Adya} {and} \bibinfo{person}{I.~L.
  Markov}.} \bibinfo{year}{2003}\natexlab{}.
\newblock \showarticletitle{{Fixed-outline floorplanning: Enabling hierarchical
  design}}.
\newblock \bibinfo{journal}{{\em IEEE Transactions on Very Large Scale
  Integration (VLSI) Systems\/}} \bibinfo{volume}{11}, \bibinfo{number}{6}
  (\bibinfo{date}{Dec} \bibinfo{year}{2003}), \bibinfo{pages}{1120--1135}.
\newblock
\showISSN{1063-8210}
\showDOI{%
\url{https://doi.org/10.1109/TVLSI.2003.817546}}


\bibitem[\protect\citeauthoryear{Alpert, Li, Sze, and Wei}{Alpert
  et~al\mbox{.}}{2013}]%
        {chuk}
\bibfield{author}{\bibinfo{person}{C.J. Alpert}, \bibinfo{person}{Z. Li},
  \bibinfo{person}{C.N. Sze}, {and} \bibinfo{person}{Y. Wei}.}
  \bibinfo{year}{2013}\natexlab{}.
\newblock \bibinfo{title}{Consideration of local routing and pin access during
  {VLSI} global routing}.
\newblock   (\bibinfo{date}{April~9} \bibinfo{year}{2013}).
\newblock
\showURL{%
\url{http://www.google.ch/patents/US8418113}}
\newblock
\shownote{{US Patent} 8,418,113.}


\bibitem[\protect\citeauthoryear{Batterywala, Shenoy, Nicholls, and
  Zhou}{Batterywala et~al\mbox{.}}{2002}]%
        {bats}
\bibfield{author}{\bibinfo{person}{S. Batterywala}, \bibinfo{person}{N.
  Shenoy}, \bibinfo{person}{W. Nicholls}, {and} \bibinfo{person}{H. Zhou}.}
  \bibinfo{year}{2002}\natexlab{}.
\newblock \showarticletitle{{Track assignment: a desirable intermediate step
  between global routing and detailed routing}}. In \bibinfo{booktitle}{{\em
  IEEE/ACM International Conference on Computer Aided Design, 2002. ICCAD
  2002.}} \bibinfo{pages}{59--66}.
\newblock
\showISSN{1092-3152}
\showDOI{%
\url{https://doi.org/10.1109/ICCAD.2002.1167514}}


\bibitem[\protect\citeauthoryear{Cao, Jing, Xiong, Hu, Feng, He, and Hong}{Cao
  et~al\mbox{.}}{2008}]%
        {zcao}
\bibfield{author}{\bibinfo{person}{Z. Cao}, \bibinfo{person}{T.~T. Jing},
  \bibinfo{person}{J. Xiong}, \bibinfo{person}{Y. Hu}, \bibinfo{person}{Z
  Feng}, \bibinfo{person}{L. He}, {and} \bibinfo{person}{X.~L. Hong}.}
  \bibinfo{year}{2008}\natexlab{}.
\newblock \showarticletitle{{Fashion: A Fast and Accurate Solution to Global
  Routing Problem}}.
\newblock \bibinfo{journal}{{\em Computer-Aided Design of Integrated Circuits
  and Systems, IEEE Transactions on\/}} \bibinfo{volume}{27},
  \bibinfo{number}{4} (\bibinfo{date}{April} \bibinfo{year}{2008}),
  \bibinfo{pages}{726--737}.
\newblock
\showISSN{0278-0070}
\showDOI{%
\url{https://doi.org/10.1109/TCAD.2008.917590}}


\bibitem[\protect\citeauthoryear{Chang, Lee, Gao, Wu, and Wang}{Chang
  et~al\mbox{.}}{2010}]%
        {ychang}
\bibfield{author}{\bibinfo{person}{Y.~J. Chang}, \bibinfo{person}{Y.~T. Lee},
  \bibinfo{person}{J.~R. Gao}, \bibinfo{person}{P.~C. Wu}, {and}
  \bibinfo{person}{T.~C. Wang}.} \bibinfo{year}{2010}\natexlab{}.
\newblock \showarticletitle{{NTHU-Route 2.0: A Robust Global Router for Modern
  Designs}}.
\newblock \bibinfo{journal}{{\em IEEE Transactions on Computer-Aided Design of
  Integrated Circuits and Systems\/}} \bibinfo{volume}{29},
  \bibinfo{number}{12} (\bibinfo{date}{Dec} \bibinfo{year}{2010}),
  \bibinfo{pages}{1931--1944}.
\newblock
\showISSN{0278-0070}
\showDOI{%
\url{https://doi.org/10.1109/TCAD.2010.2061590}}


\bibitem[\protect\citeauthoryear{Chang and Lin}{Chang and Lin}{2004}]%
        {ywchang}
\bibfield{author}{\bibinfo{person}{Y.~W. Chang} {and} \bibinfo{person}{S.~P.
  Lin}.} \bibinfo{year}{2004}\natexlab{}.
\newblock \showarticletitle{{MR: a new framework for multilevel full-chip
  routing}}.
\newblock \bibinfo{journal}{{\em Computer-Aided Design of Integrated Circuits
  and Systems, IEEE Transactions on\/}} \bibinfo{volume}{23},
  \bibinfo{number}{5} (\bibinfo{date}{May} \bibinfo{year}{2004}),
  \bibinfo{pages}{793--800}.
\newblock
\showISSN{0278-0070}
\showDOI{%
\url{https://doi.org/10.1109/TCAD.2004.826547}}


\bibitem[\protect\citeauthoryear{Chen and Chang}{Chen and Chang}{2009}]%
        {hchen2}
\bibfield{author}{\bibinfo{person}{H.~Y. Chen} {and} \bibinfo{person}{Y.~W.
  Chang}.} \bibinfo{year}{2009}\natexlab{}.
\newblock \showarticletitle{{Routing for manufacturability and reliability}}.
\newblock \bibinfo{journal}{{\em IEEE Circuits and Systems Magazine\/}}
  \bibinfo{volume}{9}, \bibinfo{number}{3} (\bibinfo{date}{Third}
  \bibinfo{year}{2009}), \bibinfo{pages}{20--31}.
\newblock
\showISSN{1531-636X}
\showDOI{%
\url{https://doi.org/10.1109/MCAS.2009.933855}}


\bibitem[\protect\citeauthoryear{Cho, Lu, Yuan, and Pan}{Cho
  et~al\mbox{.}}{2009a}]%
        {mcho1}
\bibfield{author}{\bibinfo{person}{M. Cho}, \bibinfo{person}{K. Lu},
  \bibinfo{person}{K. Yuan}, {and} \bibinfo{person}{D.~Z. Pan}.}
  \bibinfo{year}{2009}\natexlab{a}.
\newblock \showarticletitle{{BoxRouter 2.0: A Hybrid and Robust Global Router
  with Layer Assignment for Routability}}.
\newblock \bibinfo{journal}{{\em ACM Trans. Des. Autom. Electron. Syst.\/}}
  \bibinfo{volume}{14}, \bibinfo{number}{2}, Article \bibinfo{articleno}{32}
  (\bibinfo{date}{April} \bibinfo{year}{2009}), \bibinfo{numpages}{21}~pages.
\newblock
\showISSN{1084-4309}
\showDOI{%
\url{https://doi.org/10.1145/1497561.1497575}}


\bibitem[\protect\citeauthoryear{Cho, Yuan, Ban, and Pan}{Cho
  et~al\mbox{.}}{2009b}]%
        {mcho2}
\bibfield{author}{\bibinfo{person}{M. Cho}, \bibinfo{person}{K. Yuan},
  \bibinfo{person}{Y. Ban}, {and} \bibinfo{person}{D.Z. Pan}.}
  \bibinfo{year}{2009}\natexlab{b}.
\newblock \showarticletitle{{ELIAD: Efficient Lithography Aware Detailed
  Routing Algorithm With Compact and Macro Post-OPC Printability Prediction}}.
\newblock \bibinfo{journal}{{\em Computer-Aided Design of Integrated Circuits
  and Systems, IEEE Transactions on\/}} \bibinfo{volume}{28},
  \bibinfo{number}{7} (\bibinfo{date}{July} \bibinfo{year}{2009}),
  \bibinfo{pages}{1006--1016}.
\newblock
\showISSN{0278-0070}
\showDOI{%
\url{https://doi.org/10.1109/TCAD.2009.2018876}}


\bibitem[\protect\citeauthoryear{Chu and Wong}{Chu and Wong}{2008}]%
        {cchu}
\bibfield{author}{\bibinfo{person}{C. Chu} {and} \bibinfo{person}{Y.~C. Wong}.}
  \bibinfo{year}{2008}\natexlab{}.
\newblock \showarticletitle{{FLUTE: Fast Lookup Table Based Rectilinear Steiner
  Minimal Tree Algorithm for VLSI Design}}.
\newblock \bibinfo{journal}{{\em Computer-Aided Design of Integrated Circuits
  and Systems, IEEE Transactions on\/}} \bibinfo{volume}{27},
  \bibinfo{number}{1} (\bibinfo{date}{Jan} \bibinfo{year}{2008}),
  \bibinfo{pages}{70--83}.
\newblock
\showISSN{0278-0070}
\showDOI{%
\url{https://doi.org/10.1109/TCAD.2007.907068}}


\bibitem[\protect\citeauthoryear{Cong, Fang, Xie, and Zhang}{Cong
  et~al\mbox{.}}{2005}]%
        {jcong}
\bibfield{author}{\bibinfo{person}{J. Cong}, \bibinfo{person}{J. Fang},
  \bibinfo{person}{M. Xie}, {and} \bibinfo{person}{Y. Zhang}.}
  \bibinfo{year}{2005}\natexlab{}.
\newblock \showarticletitle{{MARS- a multilevel full-chip gridless routing
  system}}.
\newblock \bibinfo{journal}{{\em Computer-Aided Design of Integrated Circuits
  and Systems, IEEE Transactions on\/}} \bibinfo{volume}{24},
  \bibinfo{number}{3} (\bibinfo{date}{March} \bibinfo{year}{2005}),
  \bibinfo{pages}{382--394}.
\newblock
\showISSN{0278-0070}
\showDOI{%
\url{https://doi.org/10.1109/TCAD.2004.842803}}


\bibitem[\protect\citeauthoryear{Cormen, Leiserson, Rivest, and Stein}{Cormen
  et~al\mbox{.}}{2009}]%
        {cormen}
\bibfield{author}{\bibinfo{person}{T.~H. Cormen}, \bibinfo{person}{C.~E.
  Leiserson}, \bibinfo{person}{R.~L. Rivest}, {and} \bibinfo{person}{C.
  Stein}.} \bibinfo{year}{2009}\natexlab{}.
\newblock \bibinfo{booktitle}{{\em Introduction to Algorithms\/}
  (\bibinfo{edition}{3rd} ed.)}.
\newblock \bibinfo{publisher}{MIT Press}, \bibinfo{address}{Cambridge, MA,
  USA}.
\newblock
\showISBNx{9780262533058}


\bibitem[\protect\citeauthoryear{Ding, Gao, Yuan, and Pan}{Ding
  et~al\mbox{.}}{2011}]%
        {dding}
\bibfield{author}{\bibinfo{person}{D. Ding}, \bibinfo{person}{J.~R. Gao},
  \bibinfo{person}{K. Yuan}, {and} \bibinfo{person}{D.Z. Pan}.}
  \bibinfo{year}{2011}\natexlab{}.
\newblock \showarticletitle{{AENEID: A generic lithography-friendly detailed
  router based on post-RET data learning and hotspot detection}}. In
  \bibinfo{booktitle}{{\em Design Automation Conference (DAC), 2011 48th
  ACM/EDAC/IEEE}}. \bibinfo{pages}{795--800}.
\newblock
\showISSN{0738-100x}


\bibitem[\protect\citeauthoryear{Kar, Sur-Kolay, and Mandal}{Kar
  et~al\mbox{.}}{2013}]%
        {karb2}
\bibfield{author}{\bibinfo{person}{B. Kar}, \bibinfo{person}{S. Sur-Kolay},
  {and} \bibinfo{person}{C. Mandal}.} \bibinfo{year}{2013}\natexlab{}.
\newblock \showarticletitle{{STAIRoute}: {Global} routing using monotone
  staircase channels}. In \bibinfo{booktitle}{{\em {IEEE} Computer Society
  Annual Symposium on VLSI, {ISVLSI} 2013, Natal, Brazil, August 5-7, 2013}}.
  \bibinfo{pages}{90--95}.
\newblock


\bibitem[\protect\citeauthoryear{Kar, Sur-Kolay, and Mandal}{Kar
  et~al\mbox{.}}{2014}]%
        {karb3}
\bibfield{author}{\bibinfo{person}{B. Kar}, \bibinfo{person}{S. Sur-Kolay},
  {and} \bibinfo{person}{C. Mandal}.} \bibinfo{year}{2014}\natexlab{}.
\newblock \showarticletitle{{Global Routing Using Monotone Staircases with
  Minimal Bends}}. In \bibinfo{booktitle}{{\em 2014 27th International
  Conference on {VLSI} Design, {VLSID} 2014, Mumbai, India, January 5-9,
  2014}}. \bibinfo{pages}{369--374}.
\newblock


\bibitem[\protect\citeauthoryear{Kar, Sur-Kolay, and Mandal}{Kar
  et~al\mbox{.}}{2015}]%
        {karb4}
\bibfield{author}{\bibinfo{person}{B. Kar}, \bibinfo{person}{S. Sur-Kolay},
  {and} \bibinfo{person}{C. Mandal}.} \bibinfo{year}{2015}\natexlab{}.
\newblock \showarticletitle{{A New Method for Defining Monotone Staircases in
  {VLSI} Floorplans}}. In \bibinfo{booktitle}{{\em 2015 {IEEE} Computer Society
  Annual Symposium on VLSI, {ISVLSI} 2015, Montpellier, France, July 8-10,
  2015}}. \bibinfo{pages}{107--112}.
\newblock


\bibitem[\protect\citeauthoryear{Kar, Sur-Kolay, and Mandal}{Kar
  et~al\mbox{.}}{2016}]%
        {karb5}
\bibfield{author}{\bibinfo{person}{B. Kar}, \bibinfo{person}{S. Sur-Kolay},
  {and} \bibinfo{person}{C. Mandal}.} \bibinfo{year}{2016}\natexlab{}.
\newblock \showarticletitle{{A Novel {EPE} Aware Hybrid Global Route Planner
  after Floorplanning}}. In \bibinfo{booktitle}{{\em 29th International
  Conference on {VLSI} Design, {VLSID} 2016, Kolkata, India, January 4-8,
  2016}}. \bibinfo{pages}{595--596}.
\newblock


\bibitem[\protect\citeauthoryear{Kar, Sur-Kolay, Rangarajan, and Mandal}{Kar
  et~al\mbox{.}}{2012}]%
        {karb}
\bibfield{author}{\bibinfo{person}{B. Kar}, \bibinfo{person}{S. Sur-Kolay},
  \bibinfo{person}{S.~H. Rangarajan}, {and} \bibinfo{person}{C. Mandal}.}
  \bibinfo{year}{2012}\natexlab{}.
\newblock \showarticletitle{A Faster Hierarchical Balanced Bipartitioner for
  {VLSI} Floorplans Using {Monotone Staircase Cuts}}. In
  \bibinfo{booktitle}{{\em Progress in {VLSI} Design and Test - 16th
  International Symposium, {VDAT} 2012, Shibpur, India, July 1-4, 2012.
  Proceedings}}. \bibinfo{pages}{327--336}.
\newblock


\bibitem[\protect\citeauthoryear{Kastner, Bozorgzadeh, and Sarrafzadeh}{Kastner
  et~al\mbox{.}}{2002}]%
        {kast}
\bibfield{author}{\bibinfo{person}{R. Kastner}, \bibinfo{person}{E.
  Bozorgzadeh}, {and} \bibinfo{person}{M. Sarrafzadeh}.}
  \bibinfo{year}{2002}\natexlab{}.
\newblock \showarticletitle{{Pattern routing: use and theory for increasing
  predictability and avoiding coupling}}.
\newblock \bibinfo{journal}{{\em Computer-Aided Design of Integrated Circuits
  and Systems, IEEE Transactions on\/}} \bibinfo{volume}{21},
  \bibinfo{number}{7} (\bibinfo{date}{Jul} \bibinfo{year}{2002}),
  \bibinfo{pages}{777--790}.
\newblock
\showISSN{0278-0070}
\showDOI{%
\url{https://doi.org/10.1109/TCAD.2002.1013891}}


\bibitem[\protect\citeauthoryear{Li, Alpert, Quay, Sapatnekar, and Shi}{Li
  et~al\mbox{.}}{2007}]%
        {liz}
\bibfield{author}{\bibinfo{person}{Z. Li}, \bibinfo{person}{C.J. Alpert},
  \bibinfo{person}{S.T. Quay}, \bibinfo{person}{S. Sapatnekar}, {and}
  \bibinfo{person}{W. Shi}.} \bibinfo{year}{2007}\natexlab{}.
\newblock \showarticletitle{{Probabilistic Congestion Prediction with Partial
  Blockages}}. In \bibinfo{booktitle}{{\em Quality Electronic Design, 2007.
  ISQED '07. 8th International Symposium on}}. \bibinfo{pages}{841--846}.
\newblock
\showDOI{%
\url{https://doi.org/10.1109/ISQED.2007.124}}


\bibitem[\protect\citeauthoryear{Lin and Chu}{Lin and Chu}{2014}]%
        {tlin}
\bibfield{author}{\bibinfo{person}{T. Lin} {and} \bibinfo{person}{C. Chu}.}
  \bibinfo{year}{2014}\natexlab{}.
\newblock \showarticletitle{{POLAR 2.0: An effective routability-driven
  placer}}. In \bibinfo{booktitle}{{\em Design Automation Conference (DAC),
  2014 51st ACM/EDAC/IEEE}}. \bibinfo{pages}{1--6}.
\newblock


\bibitem[\protect\citeauthoryear{Liu, Kao, Li, and Chao}{Liu
  et~al\mbox{.}}{2013a}]%
        {wliu}
\bibfield{author}{\bibinfo{person}{W.~H. Liu}, \bibinfo{person}{W.~C. Kao},
  \bibinfo{person}{Y.~L. Li}, {and} \bibinfo{person}{K.~Y. Chao}.}
  \bibinfo{year}{2013}\natexlab{a}.
\newblock \showarticletitle{{NCTU-GR 2.0: Multithreaded Collision-Aware Global
  Routing With Bounded-Length Maze Routing}}.
\newblock \bibinfo{journal}{{\em Computer-Aided Design of Integrated Circuits
  and Systems, IEEE Transactions on\/}} \bibinfo{volume}{32},
  \bibinfo{number}{5} (\bibinfo{date}{May} \bibinfo{year}{2013}),
  \bibinfo{pages}{709--722}.
\newblock
\showISSN{0278-0070}
\showDOI{%
\url{https://doi.org/10.1109/TCAD.2012.2235124}}


\bibitem[\protect\citeauthoryear{Liu, Koh, and Li}{Liu et~al\mbox{.}}{2013b}]%
        {liuw}
\bibfield{author}{\bibinfo{person}{W.~H. Liu}, \bibinfo{person}{C.~K. Koh},
  {and} \bibinfo{person}{Y.~L. Li}.} \bibinfo{year}{2013}\natexlab{b}.
\newblock \showarticletitle{{Optimization of placement solutions for
  routability}}. In \bibinfo{booktitle}{{\em Design Automation Conference
  (DAC), 2013 50th ACM/EDAC/IEEE}}. \bibinfo{pages}{1--9}.
\newblock
\showISSN{0738-100X}


\bibitem[\protect\citeauthoryear{Majumder, Sur-Kolay, Nandy, and
  Bhattacharya}{Majumder et~al\mbox{.}}{2004}]%
        {majum1}
\bibfield{author}{\bibinfo{person}{S. Majumder}, \bibinfo{person}{S.
  Sur-Kolay}, \bibinfo{person}{S.~C. Nandy}, {and} \bibinfo{person}{B.~B.
  Bhattacharya}.} \bibinfo{year}{2004}\natexlab{}.
\newblock \showarticletitle{{On Finding a Staircase Channel with Minimum
  Crossing Nets in a VLSI Floorplan}}.
\newblock \bibinfo{journal}{{\em Journal of Circuits, Systems and Computers\/}}
  \bibinfo{volume}{13}, \bibinfo{number}{05} (\bibinfo{year}{2004}),
  \bibinfo{pages}{1019--1038}.
\newblock
\showDOI{%
\url{https://doi.org/10.1142/S0218126604001854}}
\showeprint{http://www.worldscientific.com/doi/pdf/10.1142/S0218126604001854}


\bibitem[\protect\citeauthoryear{Mitra, Yu, and Pan}{Mitra
  et~al\mbox{.}}{2005}]%
        {mitra}
\bibfield{author}{\bibinfo{person}{J. Mitra}, \bibinfo{person}{P. Yu}, {and}
  \bibinfo{person}{D.Z. Pan}.} \bibinfo{year}{2005}\natexlab{}.
\newblock \showarticletitle{{RADAR: RET-aware detailed routing using fast
  lithography simulations}}. In \bibinfo{booktitle}{{\em Design Automation
  Conference, 2005. Proceedings. 42nd}}. \bibinfo{pages}{369--372}.
\newblock
\showDOI{%
\url{https://doi.org/10.1109/DAC.2005.193836}}


\bibitem[\protect\citeauthoryear{Moffitt}{Moffitt}{2008}]%
        {moffit}
\bibfield{author}{\bibinfo{person}{M.~D. Moffitt}.}
  \bibinfo{year}{2008}\natexlab{}.
\newblock \showarticletitle{{MaizeRouter: Engineering an Effective Global
  Router}}.
\newblock \bibinfo{journal}{{\em IEEE Transactions on Computer-Aided Design of
  Integrated Circuits and Systems\/}} \bibinfo{volume}{27},
  \bibinfo{number}{11} (\bibinfo{date}{Nov} \bibinfo{year}{2008}),
  \bibinfo{pages}{2017--2026}.
\newblock
\showISSN{0278-0070}
\showDOI{%
\url{https://doi.org/10.1109/TCAD.2008.2006082}}


\bibitem[\protect\citeauthoryear{Ozdal and Wong}{Ozdal and Wong}{2009}]%
        {ozdal}
\bibfield{author}{\bibinfo{person}{M.~M. Ozdal} {and} \bibinfo{person}{M.~D.~F.
  Wong}.} \bibinfo{year}{2009}\natexlab{}.
\newblock \showarticletitle{{Archer: A History-Based Global Routing
  Algorithm}}.
\newblock \bibinfo{journal}{{\em IEEE Transactions on Computer-Aided Design of
  Integrated Circuits and Systems\/}} \bibinfo{volume}{28}, \bibinfo{number}{4}
  (\bibinfo{date}{April} \bibinfo{year}{2009}), \bibinfo{pages}{528--540}.
\newblock
\showISSN{0278-0070}
\showDOI{%
\url{https://doi.org/10.1109/TCAD.2009.2013991}}


\bibitem[\protect\citeauthoryear{Pan and Chu}{Pan and Chu}{2006}]%
        {panm1}
\bibfield{author}{\bibinfo{person}{M. Pan} {and} \bibinfo{person}{C. Chu}.}
  \bibinfo{year}{2006}\natexlab{}.
\newblock \showarticletitle{{FastRoute: A Step to Integrate Global Routing into
  Placement}}. In \bibinfo{booktitle}{{\em Computer-Aided Design, 2006. ICCAD
  '06. IEEE/ACM International Conference on}}. \bibinfo{pages}{464--471}.
\newblock
\showISSN{1092-3152}
\showDOI{%
\url{https://doi.org/10.1109/ICCAD.2006.320159}}


\bibitem[\protect\citeauthoryear{Pan and Chu}{Pan and Chu}{2007}]%
        {panm2}
\bibfield{author}{\bibinfo{person}{M. Pan} {and} \bibinfo{person}{C. Chu}.}
  \bibinfo{year}{2007}\natexlab{}.
\newblock \showarticletitle{{IPR: An Integrated Placement and Routing
  Algorithm}}. In \bibinfo{booktitle}{{\em Design Automation Conference, 2007.
  DAC '07. 44th ACM/IEEE}}. \bibinfo{pages}{59--62}.
\newblock
\showISSN{0738-100X}


\bibitem[\protect\citeauthoryear{Roy and Markov}{Roy and Markov}{2008}]%
        {royj}
\bibfield{author}{\bibinfo{person}{J.A. Roy} {and} \bibinfo{person}{I.L.
  Markov}.} \bibinfo{year}{2008}\natexlab{}.
\newblock \showarticletitle{{High-Performance Routing at the Nanometer Scale}}.
\newblock \bibinfo{journal}{{\em Computer-Aided Design of Integrated Circuits
  and Systems, IEEE Transactions on\/}} \bibinfo{volume}{27},
  \bibinfo{number}{6} (\bibinfo{date}{June} \bibinfo{year}{2008}),
  \bibinfo{pages}{1066--1077}.
\newblock
\showISSN{0278-0070}
\showDOI{%
\url{https://doi.org/10.1109/TCAD.2008.923255}}


\bibitem[\protect\citeauthoryear{Sherwani}{Sherwani}{1995}]%
        {sherw}
\bibfield{author}{\bibinfo{person}{N.~A. Sherwani}.}
  \bibinfo{year}{1995}\natexlab{}.
\newblock \bibinfo{booktitle}{{\em Algorithms for VLSI Physical Design
  Automation\/} (\bibinfo{edition}{2nd} ed.)}.
\newblock \bibinfo{publisher}{Kluwer Academic Publishers},
  \bibinfo{address}{Norwell, MA, USA}.
\newblock
\showISBNx{0792395921}


\bibitem[\protect\citeauthoryear{Sur-Kolay and Bhattacharya}{Sur-Kolay and
  Bhattacharya}{1991}]%
        {ssk}
\bibfield{author}{\bibinfo{person}{S. Sur-Kolay} {and} \bibinfo{person}{B.~B.
  Bhattacharya}.} \bibinfo{year}{1991}\natexlab{}.
\newblock \showarticletitle{{The cycle structure of channel graphs in
  nonsliceable floorplans and a unified algorithm for feasible routing order}}.
  In \bibinfo{booktitle}{{\em IEEE International Conference on Computer Design:
  VLSI in Computers and Processors}}. \bibinfo{pages}{524--527}.
\newblock
\showDOI{%
\url{https://doi.org/10.1109/ICCD.1991.139964}}


\bibitem[\protect\citeauthoryear{Viswanathan and Chu}{Viswanathan and
  Chu}{2005}]%
        {viswa}
\bibfield{author}{\bibinfo{person}{N. Viswanathan} {and} \bibinfo{person}{C.
  Chu}.} \bibinfo{year}{2005}\natexlab{}.
\newblock \showarticletitle{{FastPlace: efficient analytical placement using
  cell shifting, iterative local refinement,and a hybrid net model}}.
\newblock \bibinfo{journal}{{\em Computer-Aided Design of Integrated Circuits
  and Systems, IEEE Transactions on\/}} \bibinfo{volume}{24},
  \bibinfo{number}{5} (\bibinfo{date}{May} \bibinfo{year}{2005}),
  \bibinfo{pages}{722--733}.
\newblock
\showISSN{0278-0070}
\showDOI{%
\url{https://doi.org/10.1109/TCAD.2005.846365}}


\bibitem[\protect\citeauthoryear{Wei, Sze, Viswanathan, Li, Alpert, Reddy,
  Huber, Tellez, Keller, and Sapatnekar}{Wei et~al\mbox{.}}{2012}]%
        {weiy}
\bibfield{author}{\bibinfo{person}{Y. Wei}, \bibinfo{person}{C. Sze},
  \bibinfo{person}{N. Viswanathan}, \bibinfo{person}{Z. Li},
  \bibinfo{person}{C.J. Alpert}, \bibinfo{person}{L. Reddy},
  \bibinfo{person}{A.D. Huber}, \bibinfo{person}{G.E. Tellez},
  \bibinfo{person}{D. Keller}, {and} \bibinfo{person}{S.S. Sapatnekar}.}
  \bibinfo{year}{2012}\natexlab{}.
\newblock \showarticletitle{{GLARE: Global and local wiring aware routability
  evaluation}}. In \bibinfo{booktitle}{{\em Design Automation Conference (DAC),
  2012 49th ACM/EDAC/IEEE}}. \bibinfo{pages}{768--773}.
\newblock
\showISSN{0738-100X}


\bibitem[\protect\citeauthoryear{Westra, Bartels, and Groeneveld}{Westra
  et~al\mbox{.}}{2004}]%
        {westra}
\bibfield{author}{\bibinfo{person}{J. Westra}, \bibinfo{person}{C. Bartels},
  {and} \bibinfo{person}{P. Groeneveld}.} \bibinfo{year}{2004}\natexlab{}.
\newblock \showarticletitle{{Probabilistic Congestion Prediction}}. In
  \bibinfo{booktitle}{{\em Proceedings of the 2004 International Symposium on
  Physical Design}} {\em (\bibinfo{series}{ISPD '04})}.
  \bibinfo{publisher}{ACM}, \bibinfo{address}{New York, NY, USA},
  \bibinfo{pages}{204--209}.
\newblock
\showISBNx{1-58113-817-2}
\showDOI{%
\url{https://doi.org/10.1145/981066.981110}}


\bibitem[\protect\citeauthoryear{Xu and Chu}{Xu and Chu}{2011}]%
        {xuy2}
\bibfield{author}{\bibinfo{person}{Y. Xu} {and} \bibinfo{person}{C. Chu}.}
  \bibinfo{year}{2011}\natexlab{}.
\newblock \showarticletitle{{MGR: Multi-level global router}}. In
  \bibinfo{booktitle}{{\em Computer-Aided Design (ICCAD), 2011 IEEE/ACM
  International Conference on}}. \bibinfo{pages}{250--255}.
\newblock
\showISSN{1092-3152}
\showDOI{%
\url{https://doi.org/10.1109/ICCAD.2011.6105336}}


\bibitem[\protect\citeauthoryear{Xu, Zhang, and Chu}{Xu et~al\mbox{.}}{2009}]%
        {xuy}
\bibfield{author}{\bibinfo{person}{Y. Xu}, \bibinfo{person}{Y. Zhang}, {and}
  \bibinfo{person}{C. Chu}.} \bibinfo{year}{2009}\natexlab{}.
\newblock \showarticletitle{{FastRoute 4.0: Global router with efficient via
  minimization}}. In \bibinfo{booktitle}{{\em Design Automation Conference,
  2009. ASP-DAC 2009. Asia and South Pacific}}. \bibinfo{pages}{576--581}.
\newblock
\showDOI{%
\url{https://doi.org/10.1109/ASPDAC.2009.4796542}}


\bibitem[\protect\citeauthoryear{Zhang and Chu}{Zhang and Chu}{2012}]%
        {zhang}
\bibfield{author}{\bibinfo{person}{Y. Zhang} {and} \bibinfo{person}{C. Chu}.}
  \bibinfo{year}{2012}\natexlab{}.
\newblock \showarticletitle{{GDRouter: Interleaved global routing and detailed
  routing for ultimate routability}}. In \bibinfo{booktitle}{{\em Design
  Automation Conference (DAC), 2012 49th ACM/EDAC/IEEE}}.
  \bibinfo{pages}{597--602}.
\newblock
\showISSN{0738-100X}


\bibitem[\protect\citeauthoryear{Zhang and Chu}{Zhang and Chu}{2013}]%
        {zhang3}
\bibfield{author}{\bibinfo{person}{Y. Zhang} {and} \bibinfo{person}{C. Chu}.}
  \bibinfo{year}{2013}\natexlab{}.
\newblock \showarticletitle{{RegularRoute: An Efficient Detailed Router
  Applying Regular Routing Patterns}}.
\newblock \bibinfo{journal}{{\em Very Large Scale Integration (VLSI) Systems,
  IEEE Transactions on\/}} \bibinfo{volume}{21}, \bibinfo{number}{9}
  (\bibinfo{date}{Sept} \bibinfo{year}{2013}), \bibinfo{pages}{1655--1668}.
\newblock
\showISSN{1063-8210}
\showDOI{%
\url{https://doi.org/10.1109/TVLSI.2012.2214491}}


\bibitem[\protect\citeauthoryear{Zhang, Xu, and Chu}{Zhang
  et~al\mbox{.}}{2008}]%
        {zhang2}
\bibfield{author}{\bibinfo{person}{Y. Zhang}, \bibinfo{person}{Y. Xu}, {and}
  \bibinfo{person}{C. Chu}.} \bibinfo{year}{2008}\natexlab{}.
\newblock \showarticletitle{{FastRoute3.0: A fast and high quality global
  router based on virtual capacity}}. In \bibinfo{booktitle}{{\em
  Computer-Aided Design, 2008. ICCAD 2008. IEEE/ACM International Conference
  on}}. \bibinfo{pages}{344--349}.
\newblock
\showISSN{1092-3152}
\showDOI{%
\url{https://doi.org/10.1109/ICCAD.2008.4681596}}

\end{thebibliography}


\end{document}